\documentclass[reqno,11pt]{amsart}
\usepackage{amsmath,amsthm,amscd,amssymb,latexsym,upref,stmaryrd}
\usepackage{amsfonts,mathrsfs,dsfont}
\usepackage[top=1.5in, bottom=1.5in, left=1.4in, right=1.4in]{geometry}

\newcommand{\IR}{\mathbb{R}}
\newcommand{\ID}{\mathbb{D}}
\newcommand{\IC}{\mathbb{C}}

\newcommand{\dlab}{\langle\!\langle}
\newcommand{\drab}{\rangle\!\rangle}
\newcommand{\nM}{\nabla_{\!\!_M}}
\newcommand{\diam}{\text{diam}}

\DeclareMathOperator{\dist}{dist}
\DeclareMathOperator{\supp}{supp}

\renewcommand{\ln}{\text{\rm ln}}

\newcommand{\beq}{\begin{equation}}
\newcommand{\enq}{\end{equation}}

\let\geq\geqslant
\let\leq\leqslant

\makeatletter
\def\theequation{\@arabic\c@equation}

\allowdisplaybreaks 
\numberwithin{equation}{section}

\newtheorem{theorem}{Theorem}[section]

\newtheorem{lemma}[theorem]{Lemma}

\newtheorem{definition}[theorem]{Definition}

\newtheorem{example}[theorem]{Example}
\theoremstyle{remark}
\newtheorem{remark}[theorem]{Remark}

\begin{document}

\title[Drift-diffusion equations on  domains in $\IR^n$]{Drift-diffusion equations on domains in $\IR^d$: essential self-adjointness and stochastic completeness}

\author[G.\ Nenciu]{Gheorghe Nenciu}
\address{Gheorghe Nenciu\\Institute of Mathematics ``Simion Stoilow'' of the Romanian Academy\\ 21, Calea Grivi\c tei\\010702-Bucharest, Sector 1\\Romania}
\email{Gheorghe.Nenciu@imar.ro}

\author[I.\ Nenciu]{Irina Nenciu}
\address{Irina Nenciu\\
         Department of Mathematics, Statistics and Computer Science\\ 
         University of Illinois at Chicago\\
         851 S. Morgan Street\\
         Chicago, IL \textit{and} Institute of Mathematics ``Simion Stoilow''
     of the Romanian Academy\\ 21, Calea Grivi\c tei\\010702-Bucharest, Sector 1\\Romania}
\email{nenciu@uic.edu}

\thanks{The authors are indebted to Peter Constantin for originally suggesting that they study stochastic confinement
using the methods they had developed for quantum confinement. The research of I.N. is partly supported by NSF
grant DMS-1150427. The paper is dedicated to the 150th anniversary of the Romanian Academy.}

\begin{abstract}
We consider the problem of quantum and stochastic confinement for drift-diffusion equations on domains  $ \Omega \subset \mathbb R^d$.
We obtain various sufficient conditions on the behavior of the coefficients near the boundary of $\Omega$ which ensure the essential self-adjointness or 
 stochastic completeness  of the symmetric form of the drift-diffusion operator,  $-\frac{1}{\rho_\infty}\,\nabla\cdot \rho_\infty\mathbb D\nabla$. The proofs are based on the method developed in \cite{NN1} for  quantum confinement on bounded domains in $\mathbb R^d$. 
 In particular for stochastic confinement we combine the Liouville property with Agmon type exponential 
 estimates for weak solutions.  
\end{abstract}

\maketitle

\tableofcontents

\section{Introduction}\label{S:1}

Consider, on a  connected domain $\Omega$ in $\IR^d$, $d\geq 1$, second order partial differential operators of the form
\begin{equation*}\label{E:defn_H_V}
H=-\frac{1}{\rho_\infty}\,\nabla\cdot \rho_\infty\mathbb D\nabla +V
\end{equation*}
with $\rho_\infty$ a positive function. The problem we address in this paper is the existence, uniqueness and conservativeness of dynamics 
associated to $H$
in either a stochastic or a quantum mechanical context.
In the latter, where usually $\rho_\infty(x)=1$, $H$ describes the Hamiltonian of 
a particle with anisotropic, variable mass subject to a potential $V$. 
The corresponding dynamics is governed by Schr\"odinger's equation
$$
i \partial_t u = Hu, \quad u(\cdot,0)=u_0\, .
$$
At the heuristic level, the uniqueness of the dynamics is equivalent to the fact that, due 
to the behavior of the coefficients in $H$ near $\partial\Omega$, the particle is confined to 
$\Omega$ and never visits the boundary. Hence no boundary conditions are needed, and the coefficients determine 
completely the dynamics.
At the mathematical level, this is made rigorous by Stone's Theorem which states that there exists a unique solution 
$$
u(x,t) = e^{-iHt }u_0(x), \quad \| u(\cdot,t)\|_{L^2}= \| u_0 \|_{L^2}, \quad x\in \Omega,
$$
if and only if $H$ is essentially self-adjoint on $C_0^\infty(\Omega)$. If $\rho_\infty \not =1$, the conservativeness is in a weighted $L^2$-space, see Theorems \ref{T:ess_sa_1} and \ref{T:4} below.

In the case of stochastic particles, $H$ generates the Fokker-Planck dynamics of the system, after a suitable symmetrization  (see Section 2 for details). In this context,
$ \mathbb D(x)>0$ is the diffusion matrix,
$\rho_\infty(x)=e^{-F(x)}$ where $F(x)$ is the drift potential, and $V(x)\geq 0$ the sink potential (see for example \cite{AMTU}).
We shall only discuss the case when $V=0$, i.e., there is no absorption of the 
particle inside $\Omega$. 
To this end, we consider
\begin{equation*}
\label{E:dd_sym}
\partial_t\mu(x,t)=-\big(H_0\mu\big) (x,t)\,,\qquad \mu(\cdot,0)=\mu_0,
\end{equation*}
where
\begin{equation*}
H_0=-\frac{1}{\rho_\infty}\,\nabla\cdot \rho_\infty\mathbb D\nabla.
\end{equation*}
The question  of existence and uniqueness  of a solution to
this equation such that
\begin{equation*}
\label{E:pos_sym}
\mu(x,t)\geq 0\quad\text{for all}\,\,\, x\in\Omega\,\,\,\text{and}\,\,\, t\geq 0
\end{equation*}
and
\begin{equation*}
\int_{\Omega}\mu(x,t)\rho_\infty(x)dx=\int_{\Omega}\mu(0,x)\rho_\infty(x)dx
\quad\text{for all}\,\,\, t\geq 0
\end{equation*}
is known as the stochastic completeness problem for $H_0$.
As is well known (see, e.g., \cite{P} and references therein), one can look at this problem in the framework of 
the powerful theory of self-adjoint extensions
of the positive symmetric operator $H_0$ defined on $C_0^\infty(\Omega)$. More precisely, a self-adjoint extension $A$ of $H_0$
is called  Markovian  if its associated semigroup $P_t^A=e^{-tA}$ is a Markov process,
i.e. if it has, for $t>0$, a smooth, positive kernel $P^A(x,y;t)$ satisfying
\begin{equation*}
\label{E:defn_MarkovianExtension1}
\int_\Omega P^A(x,y;t) \rho_\infty(y)\,dy\leq 1\quad\text{for all}\,\,\, x\in\Omega\,\,\,\text{and}\,\,\, t>0\,.
\end{equation*}
If equality is attained for all $x\in\Omega$ and $t>0$,
then $A$ is called conservative. The existence  of Markovian extensions of $H_0$ is assured \cite{GM} by the fact
 that the Friedrichs
extension $H_0^F$ of $H_0$ (which for the case at hand is just the Dirichlet extension) is Markovian, and moreover
 is minimal among all Markovian extensions of $H_0$. It follows that if $H_0^F$ is conservative then it is the only Markovian
extension of $H_0$.

Thus, stochastic completeness (or, in heuristic terms, stochastic confinement) is equivalent with the fact that
$H_0^F$ is conservative, while quantum confinement is equivalent with the fact that $H_0^F$ is the only self-adjoint extension 
of
$H_0$ (i.e. $H_0$ is essentially self-adjoint). Essential self-adjointness and stochastic completeness for $H_0$ are two related but 
different problems in the theory of self-adjoint extensions of symmetric partial differential operators: it might happen that 
$H_0^F$ is the only self-adjoint extension of $H_0$ but it is not conservative  or that $H_0$ has many self-adjoint extensions
 but $H_0^F$ is conservative.

Essential self-adjointness and stochastic completeness problems for second order partial differential operators on domains in $\IR^d$
(and more generally on Riemannian manifolds) have a long and ramified history (see e.g. \cite{BMS, Br, Da, G1, GM, KSWW, RS} and references given there). 
In the one-dimensional case the situation is well understood due to the Weyl
point/circle limit theory for essential self-adjointness \cite{RS} and Feller theory for stochastic completeness \cite{Fe} (see also \cite{RoSi2}).
In the multi-dimensional case, the 
self-adjointness problem has been much 
studied, a wealth of results are available, and we send the reader to \cite{ BMS, Br, KSWW, RS} for references and a detailed account.
We only 
mention here the line of research 
 initiated by the 
celebrated paper of Wienholtz \cite{Wie} (see also \cite{J, Sim, S-H, W}), the Cordes approach \cite{Co1, Co2} and the
Brusentsev's theory \cite{Br}. 

As for the stochastic completeness problem, while many abstract criteria exist
(among them the famous Khasminskii test, and the Liouville property we shall use in this paper), 
not so many results giving explicit 
conditions on the coefficients in \eqref{E:defn_H_0} 
ensuring uniqueness or stochastic completeness  on bounded domains are known \cite{Ch}. 
A detailed study of the particular 
case $\rho_\infty(x)=1$ was only recently made \cite{RoSi1} in the framework of Dirichlet forms and capacity theory \cite{Eb, FOT, MR}. 

For the case at hand it turns out 
that there is also a beautiful 
geometric setting involved. Let $ M =(\Omega, \mathbb D^{-1})$ be the Riemannian manifold obtained by endowing 
$\Omega$ with the metric given by
\begin{equation*}
 ds^2=\sum_{j,k=1}^d\mathbb D(x)^{-1}_{j,k}dx_jdx_k.
\end{equation*}
The triple $ N=\big(\Omega, \mathbb D^{-1},\rho_\infty(x)\,dx\big)$, i.e. $ M$ equipped with the measure $\rho_\infty\, dx$,
is called  a weighted Riemannian manifold, and in this setting $H_0$ is (up to a sign) the weighted Laplace-Beltrami
operator. Now, if $ M$ is a complete Riemannian manifold, then a full answer is 
available due to abstract results on self-adjointness and stochastic completeness for Laplace-Beltrami operators on weighted
Riemannian manifolds \cite{BMS, G2, G1, GM}: $H_0$ is essentially self-adjoint and $H_0^F$ is conservative provided $\rho_\infty(x)$ 
does not increase too fast as $x \rightarrow \partial\Omega$.

The situation when  $ M$ is not complete is much more involved and the problem is far from being completely understood. 
In particular, additional conditions on $M$ seem to be necessary in order to ensure that a 
sufficiently rich family of cut-off functions with compact support exists on $M$. 
 
In a recent note \cite{NN1}, for a simple case when $M$ is not complete (more precisely for a Schr\"odinger operator  with 
repulsive potential on bounded domains in $\IR^d$ with smooth boundary),
we refined previously known results on the behavior
of the potential as $x \rightarrow \partial\Omega$ which ensure essential self-adjointness. We achieve this by combining the fundamental criterion for 
(essential) 
self-adjointness  \cite{RS} with ideas coming from Agmon's theory of exponential decay of weak solutions of partial differential equations
\cite{Ag}. 
These ideas have been expanded and used  to prove essential self-adjointness of Schr\"odinger type operators 
on graphs and non-complete Riemannian manifolds  \cite{ CT, CHT, MT, PRS}.

The main point of the present paper is that the method in  \cite{NN1} also works for the stochastic completeness of $H_0$. More 
precisely, by using the Liouville
property as criterion for stochastic completeness \cite{G1,MV} we are led to the same kind of estimates for weak solutions  as in the case of 
essential
self-adjointness.
 
We prove two types of results. In the first type, the conditions on the coefficients, while very general and elegant, are somewhat 
implicit and expressed in terms of the geometry  
and the volume element of 
$ N=\big(\Omega, \mathbb D^{-1},\rho_\infty(x)\,dx \big)$. Further work is necessary to apply them to concrete cases.
 The results in Theorem 4.1.i and Theorem 4.2.i (i.e. for $M$ is complete) are particular cases of general results 
quoted above: we include them to 
give independent and more elementary proofs. Theorem 4.1.ii is 
the generalization to our setting of Theorem 4 in \cite{NN1}.
Related results in a more general Riemannian manifold
setting have been independently obtained very recently in \cite{MT, PRS}. 
Theorem 4.2.ii is the analog for stochastic completeness of Theorem 4.1.ii.

The second type of results give explicit conditions on the coefficients in terms of the Euclidean geometry of $\Omega$
 ensuring stochastic completeness and essential 
 self-adjointness of $H_0$, respectively. The first result (Theorem \ref{T:3}) gives, for a large class of bounded domains, 
sufficient conditions on the behavior of the pair 
$\mathbb D(x),\rho_\infty (x) $ as $x \rightarrow \partial\Omega$ ensuring stochastic completeness of $H_0$. The results 
substantiate and make 
precise the heuristics that both weak diffusion ($\mathbb D(x) \rightarrow 0)$ and strong repulsive drift potential ($\rho_\infty (x)
\rightarrow 0$) enhance confinement. The interesting fact is that even when $\mathbb D(x)$ or $\rho_\infty (x)$ blows up when 
$x \rightarrow \partial\Omega$, a strong enough decay of $\rho_\infty (x) $ or $\mathbb D(x)$, respectively, still ensures the confinement. 
In this respect, Theorem \ref{T:3} is the generalization of some results
in \cite{RoSi1} where only the case $\rho_\infty \equiv 1$ was considered. The second result (Theorem \ref{T:4}) gives sufficient conditions for  
the essential 
self-adjointness of $H_0$  even in cases when $M$ is not complete. Here the main point is that even in the absence of a repulsive, confining
potential there are cases when the ``Hardy barrier'' given by the Hardy inequality is strong enough to ensure confinement. For a concrete example this fact has been proved in
\cite{Br}; see Corollary 4.1 and the remark following it. 
For similar results in Riemannian manifolds  setting see   \cite{BP, PRS}.
In Theorems \ref{T:3} and \ref{T:4}, we assume that
$\Omega$ is bounded, but an easy generalization to unbounded domains with bounded boundary is pointed out in Theorem \ref{T:5}.

The plan of the paper is as follows. In Section 2 we give some background facts about drift-diffusion equations. Since some of the 
estimates are 
somewhat technical, in Section 3 we set some notations and explain the main ideas behind the proofs in the next sections. Sections 4  and 5 contain the main results. 
Finally, in Section 6 we give some 
corollaries and examples illustrating the heuristics of quantum/stochastic confinement.

\section{Background on drift-diffusion equations}\label{S:2}

As explained above, we consider a   connected open set $\Omega\in\IR^d$, $d\geq 1$, to which we 
want to confine the stochastic and/or quantum particle. On this set, we consider a diffusion matrix
$\ID(x)=\big(\ID_{jk}(x)\big)_{1\leq j,k\leq d}$ and a drift potential $F(x)$, such that $\ID,F\in \mathcal C^\infty(\Omega)$, 
\begin{equation*}
\ID(x)=\overline{\ID(x)}=\ID^T(x)>0\quad\text{and}\quad F(x)=\overline{F(x)}\quad\text{for all}\,\,\, x\in\Omega\,.
\end{equation*}
Note that this implies that, for any compact $K\subset\Omega$, there exist constants $c_K, C_K>0$ such that:
\begin{equation*}
c_K\leq \ID(x)\leq C_K, \quad | F(x)| \leq C_K\quad\text{for all}\,\,\, x\in K\,.
\end{equation*}
We denote the (heat or particle) density function by $\rho(x,t)$, and we 
require that 
\begin{equation*}
\rho(x,t)\geq 0\quad\text{and}\quad \rho(\cdot,t)\in L^1(\Omega)
\end{equation*}
for all $x\in\Omega$ and $t\geq 0$.
The drift-diffusion equation can then be written as:
\begin{equation}\label{E:dd_original}
\partial_t\rho(x,t)=\big(\mathcal L\rho\big)(x,t)\,,\qquad \rho(\cdot,0)=\rho_0
\geq 0\,,
\end{equation}
where 
\begin{equation*}
\mathcal L\rho=\sum_{j,k=1}^n \partial_j\big(\mathbb D_{jk}(\partial_k\rho+\rho\partial_kF)\big)
=\nabla\cdot\mathbb D(\nabla\rho+\rho\nabla F)\,.
\end{equation*}

 The problem under study is the existence and uniqueness of a solution of \eqref{E:dd_original} confined in $\Omega$ 
which amounts for 
the condition
\begin{equation}
\big\|\rho(\cdot,t)\big\|_{L^1(\Omega)}=\big\|\rho_0\big\|_{L^1(\Omega)}
\quad\text{for all}\,\,\, t\geq 0\,.
\end{equation}

As we shall explain below, from the mathematical point of view it is convenient to work with the symmetrised form of $\mathcal L$.
 More precisely, define
\begin{equation}\label{E:ro_infty_defn}
\rho_\infty(x)= e^{-F(x)}\quad\text{for all}\,\,\, x\in\Omega\,,
\end{equation}
and note that, by a direct computation, $\mathcal L\rho_\infty\equiv0$.
If we set
\begin{equation*}
\mu(x,t)=\frac{\rho(x,t)}{\rho_\infty(x)}\quad\text{for all}\,\,\, x\in\Omega\,\,\,\text{and}\,\,\, t\geq 0\,,
\end{equation*}
then
\begin{equation}\label{E:dd_sym2}
\partial_t\mu(x,t)=-\big(H_0\mu\big) (x,t)\,,\qquad \mu(\cdot,0)=\mu_0\big(=\rho_0/\rho_\infty\big)
\end{equation}
where
\begin{equation}\label{E:defn_H_0}
H_0=-\frac{1}{\rho_\infty}\,\nabla\cdot \rho_\infty\mathbb D\nabla\,.
\end{equation}
To make this correlation more rigorous, define the symmetrization operator
\begin{equation*}
S\,:\, L^1_{\rho_\infty}(\Omega)\,\rightarrow\, L^1(\Omega)\,,\qquad \big(Sf\big)(x)=f(x)\rho_\infty(x)\,.
\end{equation*}
Note that $S$ is positivity-preserving: $\big(Sf\big)(x)\geq 0$ iff $f(x)\geq0$, that 
it preserves the $L^1$ norm, in the sense that:
\begin{equation*}
\big\|Sf\big\|_{L^1(\Omega)}=\big\|f\big\|_{L^1_{\rho_\infty}(\Omega)}\,,
\end{equation*}
and that it transforms the operator $\mathcal L$ into its symmetrised form: 
$$
H_0=-S^{-1}\mathcal L S=-\frac{1}{\rho_\infty}\,\nabla\cdot \rho_\infty\mathbb D\nabla
$$ 
which is symmetric on 
$$
\mathcal D(H_0)=C_0^\infty(\Omega)\subset L^2_{\rho_\infty}(\Omega)\,.
$$
Here $L^1_{\rho_\infty}(\Omega)$ and $L^2_{\rho_\infty}(\Omega)$ denote the weighted spaces as defined
in Section~\ref{S:3}.

The question of stochastic confinement then becomes the question of the existence and uniqueness of
a solution $\mu$ to \eqref{E:dd_sym2} such that 
\begin{equation}
\label{E:pos_sym2}
\mu(x,t)\geq 0\quad\text{for all}\,\,\, x\in\Omega\,\,\,\text{and}\,\,\, t\geq 0
\end{equation}
and
\begin{equation}\label{E:norm_sym}
\big\|\mu(\cdot,t)\big\|_{L^1_{\rho_\infty}(\Omega)}=\big\|\mu_0\big\|_{L^1_{\rho_\infty}(\Omega)}
\quad\text{for all}\,\,\, t\geq 0\,.
\end{equation}
The reason for using the form \eqref{E:dd_sym2},\eqref{E:defn_H_0} of the drift-diffusion
equation is precisely the fact that $H_0$ is a symmetric (unbounded) operator on $\mathcal C_0^\infty(\Omega)
\subset L^2_{\rho_\infty}(\Omega)$, and hence one can use the powerful theory of self-adjoint
extensions to investigate the question of existence and uniqueness of solutions to \eqref{E:dd_sym2} with properties
\eqref{E:pos_sym2} and \eqref{E:norm_sym}. 
The operator $H_0$ becomes a familiar object in the following geometric setting encoding its coefficients.
Let $ M :=(\Omega, \mathbb D^{-1})$ be the Riemannian manifold obtained by endowing 
$\Omega$ with the metric given by
\begin{equation}\label{E:R-metric} 
 ds^2=\sum_{j,k=1}^d\mathbb D(x)^{-1}_{j,k}dx_jdx_k.
\end{equation}
The triple $ N:=\big(\Omega, \mathbb D^{-1},\rho_\infty(x)\,dx\big)$ 
is $M$ equiped with the measure $\rho_\infty\, dx$ and 
is called  a weighted Riemannian manifold. Note that one can also view this measure with
respect to the Riemannian volume element, $v_M$, on $M$, since 
$$
\rho_\infty(x)\,dx=\rho_\infty(x)\big(\text{det}\, \mathbb D(x)\big)^{-1/2}\,dv_M(x)\,.
$$
 In this setting $H_0$ is (up to a sign) the weighted Laplace-Beltrami
operator \cite{GM}.

Note that, since $H_0\geq 0$, we know from the general theory that it has self-adjoint extensions. 
For further use, we recall  (see e.g. \cite{ GM, P}) several  facts
about the self-adjoint extensions of the closure of $H_0$
(which, by a slight but standard abuse of notation, is also denoted by $H_0$). 

\begin{definition} A self-adjoint
extension $A$ of $H_0$ is called Markovian if its associated semigroup $P_t^A=e^{-tA}$ is a Markov process,
i.e. if it has, for $t>0$, a smooth, positive kernel $P^A(x,y;t)$ satisfying
\begin{equation}\label{E:defn_MarkovianExtension}
\int_\Omega P^A(x,y;t) \rho_\infty(y)\,dy\leq 1\quad\text{for all}\,\,\, x\in\Omega\,\,\,\text{and}\,\,\, t>0\,.
\end{equation}
If equality is attained for all $x\in\Omega$ and $t>0$,
then $A$ is called conservative.
\end{definition}

Let $H_0^F$ be the Friedrichs extension of $H_0$, as defined abstractly for arbitrary positive symmetric
operators (see, e.g., \cite{Sch}, \cite{F}). For the case at hand, this is nothing more than
the Dirichlet extension of $H_0$, \cite{Sch}, \cite{GM}. It is important to note for what follows (see \cite{GM} and the 
references therein) that $H_0^F$ is a Markovian extension of $H_0$, and is minimal among all Markovian extensions 
of $H_0$, in the sense that for any Markovian extension $A$ different from $H_0^F$ of $H_0$ 
\begin{equation}
P^{H_0^F}(x,y;t)\leq P^A(x,y;t)\,,
\end{equation}
and there exist points $x_0,y_0\in\Omega$ and $t_0>0$
such that $P^{H_0^F}(x_0,y_0;t_0)< P^A(x_0,y_0;t)$.

\begin{definition}
The weighted manifold $N=\big(\Omega,\ID^{-1},\rho_\infty(x)\,dx \big)$ (or, equivalently, the operator $H_0$) is called stochastically complete 
if its Friedrichs extension $H_0^F$ is conservative.
\end{definition}

\begin{remark}
Note that it follows from the definition above that a necessary condition for the stochastic completeness
of $N$ is that $H_0$ has a unique Markovian extension, namely $H^F_0$.  Of course, by definition $H_0$ is 
essentially self-adjoint if  $H_0^F$ is its  only self-adjoint extension.
\end{remark}

\section{The approach}\label{S:3}

In this section we give the basic ideas underlying the proofs of the results in the next two sections. While most of the notation used in this paper is standard, 
we collect some of it here for ease of reference. Throughout the paper, $\Omega$ denotes a  connected
open set in $\IR^d$, $d\geq 1$. The set  $K\subset \Omega$ is said to be compact if it is compact  in $\mathbb R^d$.
Note that, since $\Omega$ is open, $K$ is compact in the subspace topology on $\Omega$ inherited from $\IR^d$ if and only if
it is compact in $\IR^d$.
For $p\geq 1$ and for a fixed weight function $w\in C(\Omega)$, 
$w(x)>0$ for every $x\in\Omega$, we consider the weighted space:
\begin{equation*}
L^p_w(\Omega)=\bigg\{f\,:\,\Omega\to\IC\,\bigg|\, f\text{ measurable and } \int_\Omega |f(x)|^pw(x)\,dx<\infty\bigg\}
\end{equation*}
with the usual norm
\begin{equation*}
\|f\|_{L^p_w(\Omega)}=\bigg(\int_\Omega |f(x)|^pw(x)\,dx\bigg)^{1/p}\,.
\end{equation*}
No lower index will be used when $w \equiv 1$.
In the case $p=2$ and weight function $\rho_\infty$, 
$L^2_{\rho_\infty}(\Omega)$ is a Hilbert space with inner product
\begin{equation*}
\dlab f,g\drab = \int_\Omega \overline{f(x)} g(x) \rho_\infty(x)\,dx\,.
\end{equation*}
We denote by $\nM$ the gradient with respect to the Riemannian structure on $M$:
\begin{equation}
\nM f(x) = \ID(x)\nabla f(x)\quad\text{for all}\,\,\, x\in\Omega\,,
\end{equation}
and by $|\cdot|_{_M}$ its norm on $M$:
\begin{equation}
|\nM f(x)|^2_{_M}=\overline{\nM f(x)}\cdot\ID(x)^{-1}\nM f(x)=\overline{\nabla f(x)}\cdot \ID(x)\nabla f(x)\,.
\end{equation}

In what follows the geometry of $M$ as well as the Euclidean geometry of $\Omega$ will enter via various distance functions.
More precisely, let $x_0$ be a point on the Riemannian manifold $M$. We then consider, for all $x \in \Omega$,  the distance to a point, $d_M(x)=\dist_M(x,x_0)$ and (see e.g. \cite{Ag} for details) the distance to the "boundary":
\begin{equation*}
\delta_M(x)=\text{dist}_M(x,\partial\Omega) := \sup_K \{ d_M(x, \Omega\setminus K) |\text{ $K$ is a compact subset of } \Omega\}.
\end{equation*}
  These functions are distance functions i.e they have the property that
\begin{equation}
\big|d_M(x)-d_M(x')\big|\leq \text{dist}_M(x,x')\quad\text{and}\quad \big|\delta_M(x)-\delta_M(x')\big|\leq \text{dist}_M(x,x')
\end{equation}
for all $x,x'\in\Omega$. Hence, by Rademacher's theorem,  the two functions are (a.e.) differentiable with:
\begin{equation}\label{E:2.1.4}
\big|\nM d_M\big|_{_M}\leq1\quad\text{and}\quad\big|\nM \delta_M\big|_{_M}\leq1\,.
\end{equation}

Note that, from the geometric point of view, $M$ is in exactly one of the following three cases:
\begin{itemize}
\item[\textbf{(C1)}] $M$ is geodesically complete, which by the Hopf-Rinow theorem is equivalent to $\delta_M(x)=\infty$ for every $x\in\Omega$;\label{geodesic_case1}
\item[\textbf{(C2)}] \label{geodesic_case2} 
$M$ is geodesically incomplete and 
\begin{equation*}
\text{diam}(M)=\sup\big\{\!\dist_M(x,y)\,\big|\,x,y\in\Omega\big\}<\infty\,;
\end{equation*}
\item[\textbf{(C3)}] \label{geodesic_case3}
$M$ is geodesically incomplete and 
\begin{equation*}
\text{diam}(M)=\sup\big\{\!\dist_M(x,y)\,\big|\,x,y\in\Omega\big\}=\infty\,.
\end{equation*}
\end{itemize}
For $A,B \subset \mathbb R^d$ we denote $d(A,B) = \inf_{x \in A,y\in B}|x-y|$.
We shall also employ the standard Euclidean distance  to the boundary, $\delta(x)$, on  $\Omega$ :
$\delta(x)= \inf_{y \in \partial \Omega}|x-y|$.

There are many abstract criteria ensuring stochastic completeness for $H_0$ as given in the previous section, 
respectively essential self-adjointness for 
\begin{equation}\label{E:defnH_2}
H=H_0+V\,,\quad\text{where } V(x) \geq 0\,,\,\,\, V \in L^{\infty}_{loc}(\Omega) 
\end{equation}
(\cite{G2,G1,GM,MV,RS, Sch} and
 references therein). We shall use two such criteria which, in spite of the fact that they answer two different questions, are very similar.

The first one is  the so-called Liouville property: 
\begin{theorem}\label{SC}\cite[\S1.3]{G1}
The following statements are equivalent:
\begin{itemize}
\item[(i)] $N=\big(\Omega,\ID^{-1},\rho_\infty(x)\,dx \big)$ is stochastically complete;
\item[(ii)] There exists $E<0$ such that the only bounded, classical solution of $(H_0-E)\psi_E=0$ is $\psi_E=0$;
\item[(iii)] For any $E<0$ it holds that the only bounded classical solution of $(H_0-E)\psi_E=0$ is $\psi_E=0$.
\end{itemize}
\end{theorem}

The second criterion is  the basic criterion for essential self-adjointness:
\begin{theorem}\label{ES}\cite{RS}
The following statements are equivalent:
\begin{itemize}
\item[(i)] $H$ is essentially self-adjoint;
\item[(ii)] There exists $E<0$ such that the only solution $\psi_E\in L^2_{\rho_\infty}(\Omega)$ of
\begin{equation}
\dlab (H-E)\varphi,\psi_E\drab=0 \quad\text{for all }\varphi\in C_0^\infty(\Omega)
\end{equation}
is $\psi=0$;
\item[(iii)] For any $E<0$ it holds that the only solution $\psi_E\in L^2_{\rho_\infty}(\Omega)$ of
\begin{equation}
\dlab (H-E)\varphi,\psi_E\drab=0 \quad\text{for all }\varphi\in C_0^\infty(\Omega)
\end{equation}
is $\psi_E=0$.
\end{itemize}
\end{theorem}

We would like to stress that the only difference between the two criteria is the fact 
that $\psi_E$ belongs to a different space:  ${L^{\infty}(\Omega)}$ and ${L^2_{\rho_\infty}(\Omega)}$, respectively. 
 As a consequence, a method of verifying one criterion is 
likely to also work for the other.
Notice that for any $\varphi\in C_0^\infty(\Omega)$, we have
\begin{equation}
\dlab \varphi, (H_0-E)\psi\drab=\dlab (H_0-E)\varphi,\psi\drab=0\,,  
\end{equation}
and so naturally classical solutions are also weak solutions.

An important technical ingredient of our analysis is contained in the following two lemmas stated for $H$ and valid in 
particular also for $H_0$. Let $ h [\phi, \psi]  $ be the quadratic form associated to  $H$ i.e. 
for $ \phi, \psi \in C_0^\infty(\Omega)$: 
\begin{equation}
h[\phi,\psi]=\int_\Omega \left(\overline{\nabla \phi(x)}\cdot \ID(x)\nabla \psi(x)
+ \overline{\phi(x)} V(x)\psi(x)\right) \rho_\infty(x)\,dx.
\end{equation}
As a quadratic form associated to a symmetric positive operator, $h$ is closable \cite{F}  and with a slight abuse of notation we shall denote
by the same symbol its closure, with domain $Q(h)$. Finally, we denote by $\mathcal L_0(\Omega)$ the set of Lipschitz continuous functions with
compact support inside $\Omega$.

\begin{lemma}[Localization Lemma]\label{L:L}
Suppose $E\in\IR$, $\psi_E$ is a weak solution of
\begin{equation}
(H-E)\psi_E=0\,,
\end{equation}
and $f=\bar f\in \mathcal L_0(\Omega)$. Then $f\psi_E \in Q(h)$ and
\begin{equation}\label{E:quad_ident}
h[f\psi_E,f\psi_E]-E\dlab f\psi_E,f\psi_E\drab =\dlab \psi_E,|\nM f|_{_M}^2 \psi_E\drab\,.
\end{equation}
\end{lemma}

As in the case of many other fundamental identities, \eqref{E:quad_ident} has been rediscovered
many times (see e.g. \cite{KSWW, BMS, Ag}). In particular, in the classical paper \cite{Wie} it appeared as an essential ingredient
in the study of essential self-adjointness (see also \cite{J,S-H,W,Sim}). It is also the key
ingredient of Agmon's \cite{Ag} exponential estimates for generalized eigenfunctions. In all applications of 
\eqref{E:quad_ident} the key step is the estimate from above of the right-hand side, and here
the choice of function $f$ is crucial. Agmon's insight here was that the choice of $f$ must be linked
to a certain Riemannian metric (aka the Agmon distance) dictated by the coefficients of the PDE involved. 
The proof of this lemma is standard (see the papers quoted above and \cite{NN1} for the proof in the most elementary setting)  and we will not repeat it here. 
Note, however,
that for a function $\psi$ with compact support, $\psi\in W^{1,2}(N)$ if and only if $\psi\in W^{1,2}(\Omega)$, where
$W^{1,2}(\Omega)$ is the usual (Euclidean and not weighted) Sobolev space and that for 
$B(x_0,r):= \{x\in \Omega| |x-x_0| \leq r\}$ the set of uniformly Lipschitz continuous functions coincides with 
 $W^{1,\infty}(B(x_0,r))$.

Adapting Agmon's insight to our setting we use Lemma \ref{L:L} to obtain the following:
\begin{lemma}[Basic Inequality]\label{L:B}
Assume that there exist $E_0>-\infty$ and a function $B\,:\,\Omega\,\rightarrow\,[0,\infty)$ such that
for all $\varphi\in C_0^\infty(\Omega)$, it holds that:
\begin{equation}\label{E:h}
\dlab \varphi, (H-E_0)\varphi\drab \geq \int_\Omega |\varphi(x)|^2 B(x)\rho_\infty(x)\,dx\,.
\end{equation}
Let $\psi_E$ be a weak solution of $\big(H-E\big)\psi_E=0$ for some $E< E_0$,
and let $g$ be a real-valued, Lipschitz continuous function on $\Omega$ satisfying:
\begin{equation}\label{E:3.1.4}
\big|\nM g(x)\big|_{_M}^2\leq  B(x)+\frac{|E-E_0|}{2}\,.
\end{equation}
Then
\begin{equation}\label{E:3.1.5}
\dlab\psi_E , f^2\psi_E\drab\leq \frac{2}{|E-E_0|} \dlab\psi_E,|m|\psi_E\drab\,,
\end{equation}
where $f=e^g\phi$ with $\phi\,:\,\Omega\rightarrow[0,1]$, $\phi\in\mathcal L_0(\Omega)$, and
\begin{equation}\label{E:3.1.7}
m=e^{2g}\big(2\phi(\nM g)\cdot_{_M}(\nM\phi)+|\nM \phi|_{_M}^2 \big)\,.
\end{equation}
\end{lemma}

\begin{proof}
Recall that $Q(h)$ is the domain of $h$. By density $f\psi_E \in Q(h)$, and from \eqref{E:quad_ident}, \eqref{E:h}, by adding and subtracting $E_0\dlab f\psi_E,f\psi_E\drab$, one has:
\begin{equation*}
\frac{|E-E_0|}{2}\dlab f\psi_E,f\psi_E\drab \leq \dlab \psi_E,|\nM f|_{_M}^2 \psi_E\drab - \dlab f\psi_E,  \left( B(x)+\frac{|E-E_0|}{2} \right) f\psi_E\drab
\end{equation*}
Further, inserting   $f=e^g\phi$  in the first term of the r.h.s. of the above inequality and using \eqref{E:3.1.4} one obtains
\begin{equation*}
\frac{|E-E_0|}{2}\dlab f\psi_E,f\psi_E\drab \leq  \dlab\psi_E, |m| \psi_E\drab
\end{equation*}
which gives \eqref{E:3.1.5}.
\end{proof}

We use \eqref{E:3.1.5}, \eqref{E:3.1.7} to prove that $\psi_E =0$. The simplest scenario to achieve that (in some cases a 
slightly more sophisticated scenario 
is required to obtain better
results; see proofs of Theorems \ref{T:ess_sa_1} and \ref{T:4} below and  \cite{NN1}) runs as follows. Fix a compact $K$ in $\Omega$. 
Then choose a sequence of cut-off functions $\phi_n$ such that
$\phi_n|_K=1$. If in the limit $n\rightarrow \infty$ the r.h.s. of \eqref{E:3.1.5} goes to zero, then one obtains that $\psi_E|_K=0$. 
Doing this for a 
sequence of compacts exhausting $\Omega$, one arrives at the needed result. Notice that as the compact $K$ gets larger and 
larger, the condition 
$\phi_n|_K=1$ together with the condition that $\phi_n$ has compact support might imply the blow up  of $|\nM \phi|_{_M}$, which has to be 
compensated for by the decay of $e^{2g}\rho_\infty$.
In each particular case the main point is an appropriate choice of $f= e^g\phi$ in order to make the r.h.s. of \eqref{E:3.1.5} as 
small as possible.  In what follows, $g$ and $\phi$ will be chosen to be functions of the distance functions defined above.
For example, when we use the geometry of $M$, $g$ will  be chosen to be either a linear function of $d_M(x)$ or of the form
$G(\delta_M(x))$ where $G(t)$ satisfies the following condition (see \cite{NN1,NN2} for variants and examples):
\begin{definition}[$\Sigma$]
Let $G:(0,\infty)\rightarrow \mathbb R_-$ be locally Lipschitz continuous, satisfying:

$\Sigma_1$. There exists $1 \geq a_0 >0$ such that (a.e)
\begin{equation}\label{E:sigma1}
\frac{1}{t} \geq G'(t) \geq 0, \quad \text{for} \quad t \in (0,a_0); \quad
G'(t)=G(t)=0, \quad \text{for} \quad t >a_0
\end{equation}

$\Sigma_2$.
\begin{equation}\label{E:sigma2}
\sum_{m=1}^{\infty}4^{-m}e^{-2G(2^{-m}a_0)} =\infty.
\end{equation}
\end{definition}

In order to ensure that the functions $\phi$ have compact support we need the following assumption.
For $r,R \in (0, \infty)$ we denote by $S_{r,R}$ the level sets:
\begin{equation}
S_{r,R}=\big\{x\in\Omega\,\big|\, \delta_M(x)\geq r\text{ and }d_M(x)\leq R\big\}.
\end{equation}
 Notice that if $M$ is complete then for all  $R, r \in (0, \infty)$,
$$S_{r,R} =S_{R}:=\big\{x\in\Omega\,\big|\, d_M(x)\leq R\big\}\,,$$ and if $\text{diam}(M) <\infty$, for $R >\text{diam}(M)$,
$$S_{r,R}=S_{r} := \big\{x\in\Omega\,\big|\, \delta_M(x)\geq r\}\,.$$

\noindent
\textbf{Assumption} \textbf{(A).} 
\textit{For all  $R, r \in (0, \infty)$, $S_{r,R} \subset \Omega $ are compacts.}
\hspace{3mm}

We note that in the Euclidean case, i.e. if
$\mathbb D$ is the identity matrix, or more generally if $H_0$ is strongly elliptic, i.e. 
$ 0< c \leq  \mathbb D(x) \leq C< \infty$ uniformly for $x \in \Omega\,,$ this assumption holds true. 

Although in developing the approach outlined above we were inspired by Agmon theory \cite{Ag}, it has to 
be remarked that it has its roots in the approach initiated by Wienholtz  \cite{Wie} and further developed in \cite{J,S-H,W,Sim}. 

While in this paper we restricted ourselves to applying this approach to prove essential self-adjointness 
and stochastic completeness for drift-diffusion equation in domains of $\mathbb R^d$, it is an interesting program to see to what extent 
one can use the same strategy 
for partial differential operators on domains on more general manifolds. As mentioned in the Introduction, this program has already been partially carried out 
for essential self-adjointness of Schr\"odinger type operators on graphs and non-complete Riemannian manifolds. 
By adding the necessary technical steps, one can also 
reformulate and extend the results concerning stochastic completeness to a more general geometric setting, as done
in \cite{G2, G3, G1, GM} for complete Riemannian manifolds. 
As a future development, it would be very interesting to extend these results to, for example, the case of the Laplace-Beltrami operator 
in almost Riemannian geometry \cite{BL}.

\section{Criteria in terms of the Riemannian metric}\label{S:4}

\begin{theorem}\label{T:ess_sa_1}
Assume that:
\begin{itemize}
\item[(i)] $M=(\Omega,\mathbb D^{-1})$ is geodesically complete, i.e. $
\delta_M(x)\equiv\infty$,\\ 
or
\item[(ii)] Assumption \textbf{(A)} holds true and there exists $E<0$ such that for all $\varphi\in C_0^\infty(\Omega)$,
\begin{equation}\label{E:BL}
\dlab\varphi,(H-E)\varphi\drab -\int_\Omega |\nM g|_{_M}^2|\varphi|^2 \rho_\infty\,dx\geq 
\dlab\varphi,\varphi\drab
\end{equation}
where $H=H_0+V$ and $g(x)=G\big(\delta_M(x)\big)$ for some $G$ satisfying condition {\rm (}$\Sigma${\rm )}. 
\end{itemize}
Then $H$ is essentially self-adjoint in $L^2_{\rho_\infty}(\Omega)$.
\end{theorem}

\begin{proof}
As discussed above, we note that 
\begin{equation*}
\diam(M)<\infty \qquad \text{iff} \qquad \sup_{x\in\Omega} d_M(x)<\infty\,,
\end{equation*}
and that $M$ geodesically complete implies $\diam(M)=\infty$ (though the converse is not true). 

We describe here the class of cut-off functions $\phi(x)$ which will be used in the application
of the Basic Inequality Lemma \ref{L:B}.  The fact that they have compact support is ensured by Assumption (\textbf{A}).
 We focus on the case (C3) (i.e. M is not complete and its diameter is infinite, see Section 3), in which 
hypothesis (ii) of the theorem must hold. Let $(r_j)_{j\geq 1}$ be a strictly decreasing sequence and
$(R_j)_{j\geq 1}$ a strictly increasing sequence, such that
\begin{equation}\label{E:3.2.1}
r_1\leq1\leq R_1\,,\quad \lim_{j\to\infty} r_j =0\quad\text{and}\quad \lim_{j\to\infty}R_j=\infty\,.
\end{equation}
Using these sequences, define the sets:
\begin{equation}\label{E:3.2.5}
\Omega_{r_j,R_j}=\big\{x\in\Omega\,\big|\, \delta_M(x)>r_j\text{ and }d_M(x)<R_j\big\}
\quad\text{and}\quad A_j=\Omega_{r_{j+1},R_{j+1}}\setminus \overline{\Omega_{r_j,R_j}}\,.
\end{equation}
Note that $A_j\cap A_m=\emptyset$ for all $j\neq m$,
and that for any compact set $K\subset\Omega$ there exists $j(K)\geq1$ such that
\begin{equation}
K\subset \Omega_{r_j,R_j}\qquad\text{for all } j> j(K)\,.
\end{equation}

Furthermore, define two sequences of smooth functions, 
\begin{equation}
k_j,\ell_j\,:\,(0,\infty)\rightarrow [0,1]
\end{equation}
such that:
\begin{equation}\label{E:3.2.2,3}
k_j(t)=
\begin{cases}
0\,, &\quad t<r_{j+1}\\
1\,, &\quad t>r_{j}
\end{cases}\qquad
\ell_j(t)=
\begin{cases}
0\,, &\quad t>R_{j+1}\\
1\,, &\quad t<R_{j}
\end{cases}
\end{equation}
\begin{equation}
\text{supp}(k_j')\subseteq(r_{j+1},r_j)\,\qquad \text{supp}(\ell_j')\subseteq(R_{j},R_{j+1})\,,
\end{equation}
and
\begin{equation}\label{E:lkderiv}
\big|k_j'(t)\big|\leq \frac{2}{r_j-r_{j+1}}\,,\qquad\big|\ell_j'(t)\big|\leq \frac{2}{R_{j+1}-R_{j}}\,.
\end{equation}
Denote $\mathcal K_j(x)=k_j(\delta_M(x))$, $\mathcal L_j(x)=\ell_j(d_M(x))$, and
\begin{equation}\label{E:3.2.8}
\phi_j(x)=\mathcal K_j(x)\mathcal L_j(x)\,.
\end{equation}

By construction, 
\begin{equation}\label{E:3.2.9}
\phi_j(x)=
\begin{cases}
1 & \quad x\in\Omega_{r_j,R_j}\\
0 &\quad x\notin \Omega_{r_{j+1},R_{j+1}}\,,
\end{cases}
\end{equation}
and
\begin{equation}\label{E:3.2.10}
\text{supp}\big(\nM\phi_j\big)\subset A_j\,.
\end{equation}
Using \eqref{E:2.1.4},  \eqref{E:3.2.2,3} and \eqref{E:lkderiv}, we can also conclude that
\begin{equation}\label{E:3.2.12}
\big|\nM \mathcal K_j(x)\big|_{_M}\leq\frac{2}{r_j-r_{j+1}}\quad\text{and}
\quad \nM\mathcal K_j(x)=0\text{ for }\delta_M(x)\notin\big(r_{j+1},r_j\big)\,,
\end{equation}
and
\begin{equation}\label{E:3.2.13}
\big|\nM\mathcal L_j(x)\big|_{_M}\leq\frac{2}{R_{j+1}-R_j}\quad\text{and}
\quad \nM\mathcal L_j(x)=0\text{ for }d_M(x)\notin\big(R_{j},R_{j+1}\big)\,.
\end{equation}

Let
\begin{equation}\label{E:3.2.15}
g(x)=G\big(\delta_M(x)\big)
\end{equation}
be as in hypothesis (ii) in the theorem (recall that we are working in case (C3), in which this hypothesis must hold).
We use this function $g$ in Lemma~\ref{L:B}, as well as $\phi(x)=\phi_j(x)$ as given in \eqref{E:3.2.8}. 
From \eqref{E:sigma1} we then find that
\begin{equation}\label{E:nablag}
\big|\nM g(x)\big|_{_M}\leq\frac{1}{\delta_M(x)}
\end{equation}
and, for $x\in\Omega$ such that $\delta_M(x)>r_j$,
\begin{equation}\label{E:3.2.17}
g(x)-G(r_j)\leq \int_{r_j}^{\delta_M(x)} G'(t)\,dt\leq \int_{r_j}^{1} G'(t)\,dt \leq \ln\big(1/r_j\big)\,.
\end{equation}
Hence for $\delta_M(x)>r_j$
\begin{equation}\label{E:3.2.18}
e^{2g(x)}\leq \frac{1}{r_j^2}\cdot e^{2G(r_j)}\,,
\end{equation}
while for $\delta_M(x)\leq r_j$,
\begin{equation}\label{E:3.2.19}
e^{2g(x)}\leq e^{2G(r_j)}\,.
\end{equation}
From \eqref{E:BL} the conditions \eqref{E:h} and \eqref{E:3.1.4} in Lemma \ref{L:B} hold true for  $E_0=E+1$ and $B(x)=|\nM g|_{_M}^2$.

We now use all of the above to estimate the right-hand side of \eqref{E:3.1.5}. We take 
\begin{equation}\label{E:3.2.20}
r_j=\frac{r_1}{2^{j-1}}\,,\quad r_1<a_0\,\quad\text{and}\quad R_j =\frac{1}{r_j}
\end{equation}
in which case:
\begin{equation}\label{E:3.2.21}
r_j-r_{j+1}=r_{j+1}\quad\text{and}\quad R_{j+1}-R_j=R_j\,.
\end{equation}
Therefore, for $x\in  A_j$ (see \eqref{E:3.1.7} and \eqref{E:nablag}),
\begin{equation}\label{E:3.2.22}
\begin{aligned}
|m_j(x)| 
&\leq e^{2g(x)}\bigg[\frac{2}{\delta_M(x)}\Big(\big|\nM \mathcal K_j(x)\big|_{_M} \ell_j(d_M(x))
 +k_j(\delta_M(x))\big|\nM\mathcal L_j(x)\big|_{_M}\Big)\\
&\qquad\qquad +\Big(\big|\nM \mathcal K_j(x)\big|_{_M} \ell_j(d_M(x))
 +k_j(\delta_M(x))\big|\nM\mathcal L_j(x)\big|_{_M}\Big)^2\bigg] 
\end{aligned}
\end{equation}

From its definition, we can write $A_j=A_j^1\cup A_j^2$ with
\begin{equation}
A_j^1=\big\{x\in\Omega\,\big|\, r_{j+1}<\delta_M(x)<r_j,\,\,\, d_M(x)<R_{j+1}\big\}
\end{equation}
and
\begin{equation}
A_j^2=\big\{x\in\Omega\,\big|\, \delta_M(x)\geq r_j,\,\,\, R_j<d_M(x)<R_{j+1}\big\}\,.
\end{equation}
Let $x\in A_j^1$, and use \eqref{E:3.2.12}, \eqref{E:3.2.13},\eqref{E:3.2.19}, \eqref{E:3.2.20}, \eqref{E:3.2.21}, as well as 
the fact that $r_1<a_0\leq1$ to obtain that:
\begin{equation}\label{E:3.2.23}
\big|m_j(x)\big| \leq 4e^{2G(r_j)}\left[\frac{1}{r_{j+1}}\left(\frac{1}{r_{j+1}}+r_j\right)+\left(\frac{1}{r_{j+1}}+r_j\right)^2\right]
 \leq\text{const. } \frac{e^{2G(r_j)}}{r_j^2}\,.
\end{equation}
For $x\in A^2_j$, recall that $k_j'(\delta_M(x))=0$, so using \eqref{E:3.2.13}, \eqref{E:3.2.18}, \eqref{E:3.2.20},
and \eqref{E:3.2.21} we obtain that
\begin{equation}\label{E:3.2.19'}
\big|m_j(x)\big|\leq \text{const. } \frac{e^{2G(r_j)}}{r_j^2}\left[\frac{2}{r_j}r_j+r_j^2\right]
\leq \text{const. } \frac{e^{2G(r_j)}}{r_j^2}\,.
\end{equation}
Summarizing, we have shown that for $x\in A_j$ we have
\begin{equation}
\big|m_j(x)\big|\leq \text{const. } \frac{e^{2G(r_j)}}{r_j^2}\,,
\end{equation}
and hence
\begin{equation}\label{E:3.2.20'}
\dlab\psi_E,|m_j|\psi_E\drab \leq \text{const. } \frac{e^{2G(r_j)}}{r_j^2}\int_{A_j} \big|\psi_E(x)\big|^2 \rho_\infty(x)\,dx\,.
\end{equation}

Let now $K$ be a compact set, $K\subset\Omega$. Since $\Omega_{r,1/r}$ exhaust $\Omega$
as $r\to0$, there exists $r_0\in(0,2a_0)$ such that $K\subset \Omega_{r_0,1/r_0}$. Take then $r_1=r_0/2$.
Since $\Omega_{r_0,1/r_0}\subset \Omega_{r_1,1/r_1}$, we know that $\phi_j(x)=1$ for all
$x\in K$ and all $j\geq 1$. Furthermore, since $G'\geq 0$, we have
\begin{equation}
\inf_{x\in K} e^{2g(x)}\geq e^{2G(r_0)}\,,
\end{equation}
which then gives:
\begin{equation}\label{E:3.2.21'}
\dlab\psi_E,f_j^2\psi_E\drab\geq \int_K\big|\psi_E(x)\big|^2 f_j(x)^2\rho_\infty(x)\,dx 
\geq e^{2G(r_0)}\int_K \big|\psi_E(x)\big|^2 \rho_\infty(x)\,dx\,.
\end{equation}
Together with \eqref{E:3.1.5} and \eqref{E:3.2.20'}, this leads to 
\begin{equation}\label{E:3.2.22'}
\text{const. }|E|e^{2G(r_0)} r_j^2 e^{-2G(r_j)} \int_K\big|\psi_E(x)\big|^2 \rho_\infty(x)\,dx \leq 
\int_{\bar A_j} \big|\psi_E(x)\big|^2 \rho_\infty(x)\,dx\,.
\end{equation} 
Summing over $j\geq 1$ and recalling that $r_j=r_0 2^{-j}$ yields:
\begin{equation}\label{E:3.2.23'}
\text{const.}(|E|,r_0)\left(\sum_{j=1}^\infty 4^{-j} e^{-2G(2^{-j}r_0)}\right) 
\int_K \big|\psi_E(x)\big|^2 \rho_\infty(x)\,dx\leq \|\psi_E\|_{L^2_{\rho_\infty}}^2\,.
\end{equation}
Now, for any integer $N$, condition ($\Sigma_2$) implies that
\begin{equation}
\sum_{j=1}^\infty 4^{-j} e^{-2G(a_0 2^{-N}2^{-j})}=\infty\,,
\end{equation}
so if we take $r_0=a_0 2^{-N}$ with $N$ sufficiently large to ensure $K\subset \Omega_{r_0,1/r_0}$,
we obtain
\begin{equation}
\sum_{j=1}^\infty 4^{-j} e^{-2G(2^{-j}r_0)}=\infty\,.
\end{equation}
Combining this with \eqref{E:3.2.23'} yields
\begin{equation}
\int_K \big|\psi_E(x)\big|^2 \rho_\infty(x)\,dx=0\,,
\end{equation}
which in turn, given that $K$ is an arbitrary compact subset of $\Omega$ and $\rho_\infty>0$ on $\Omega$, implies that 
$\psi_E=0$. This completes the proof in the case (C3).

The proof in the cases (C1) and (C2) is similar to the one above. Note that
$\mathcal K_j$ cannot be defined in the case (C1), and that in the 
case (C2) and for $j$ sufficiently large,
we have $\mathcal L_j\equiv 1$ (and thus have no cut-off effect). 
It is therefore natural to make the following (simpler) choices for $f$:\\
In the case (C1), take $f_j(x)=\mathcal L_j(x)$ and $\Omega_{R_j}=\{x\in\Omega\,|\,d_M(x)<R_j\}$. Recall that in 
this case we do not (necessarily) have hypothesis (ii), and neither a function $g$ as used above (since it is not needed in 
this case, and we apply Lemma~\ref{L:B} with $g\equiv 0$). Note however that
the $\Omega_{R_j}$'s exhaust $\Omega$ as $j\to\infty$.\\
In the case (C2) (in which hypothesis (ii) must again hold), take $f_j(x)=e^{g(x)}\mathcal K_j(x)$ and 
$\Omega_{r_j}=\{x\in\Omega\,|\,\delta_M(x)>r_j\}$. 
Again, the $\Omega_{r_j}$'s exhaust $\Omega$ as $j\to\infty$.

In both of these cases, the proof proceeds exactly as above, with many of the intermediate inequalities becoming
somewhat simplified.
\end{proof}

In the case of stochastic confinement, the balance between the effects of the geometry of $\Omega$ and those
of the confinement can be achieved in several ways. We show several possible results, as explained in
the introduction.

\begin{theorem}\label{T:2}
Assume that $\rho_\infty(x)=\omega(x)\tilde\rho_\infty(x)$ with $\tilde\rho_\infty\in L^1(\Omega,\rho_\infty\, dx)$.
Let $(r_k)_{k\geq1}$ be a decreasing sequence going to 0, and $R_{k+1}:=R_k+1$.

Then $H_0$
is stochastically complete in all the following cases:
\begin{itemize}
\item[(i)] $M$ is geodesically complete and there exists $\alpha<\infty$ such that 
\begin{equation}\label{E:2.2.4}
\omega_{k,\infty}:=\sup_{R_k\leq d_M(x)\leq R_{k+1}} \omega(x)\leq e^{\alpha R_k}\,.
\end{equation}
\item[(ii)] $M$ is geodesically incomplete,  Assumption \textbf{(A)} holds true,    $\text{diam}(M)<\infty$, and
\begin{equation}\label{E:2.2.6}
\sum_{k=1}^\infty \frac{(r_k-r_{k+1})^2}{\omega_k} =\infty\,\quad
\text{where}\quad
\omega_k:= \sup_{r_{k+1}\leq \delta_M(x)\leq r_{k}} \omega(x)\,.
\end{equation}
\item[(iii)] $M$ is geodesically incomplete, Assumption \textbf{(A)} holds true, $\text{diam}(M)=\infty$, and both conditions \eqref{E:2.2.4}, \eqref{E:2.2.6} hold.
\end{itemize}
\end{theorem}

\begin{proof}
Again we give in detail only the proof for case (C3), which now corresponds exactly with hypothesis (iii) holding. 
We will keep the same choices for $\mathcal K_j$, $\mathcal L_j$
and $\Omega_{r_j,R_j}$ as in the proof of Theorem~\ref{T:ess_sa_1}, but instead of \eqref{E:3.2.20} we take:
\begin{equation}\label{E:3.3.1}
r_j\to 0\quad\text{and}\quad R_{j+1}=R_j+1\,,
\end{equation}
and we choose
\begin{equation}\label{E:3.3.2}
g(x)=-\frac{\alpha}{2} d_M(x)\quad\Longrightarrow\quad \big|\nM g(x)\big|_{_M}\leq \frac{\alpha}{2}\,.
\end{equation}

From \eqref{E:3.1.7}, we obtain in this case  :
\begin{equation}\label{E:3.3.4}
|m_j(x)|\leq e^{-\alpha d_M(x)} \mathcal A_j(x)\cdot\big[\alpha+\mathcal A_j(x)\big]\,,
\end{equation}
with
\begin{equation}\label{E:3.3.5}
\mathcal A_j(x)=\big|\nM \mathcal K_j(x)\big|_{_M} \ell_j(d_M(x))
 +k_j(\delta_M(x))\big|\nM\mathcal L_j(x)\big|_{_M}.
\end{equation}
For $x\in A_j^1$, we find that (see \eqref{E:3.2.12}, \eqref{E:3.2.13} and \eqref{E:3.3.1})\begin{equation}\label{E:3.3.6}
\mathcal A_j(x)\leq 2\,\left(1+\frac{1}{r_j-r_{j+1}}\right)\,,
\end{equation}
while for $x\in A_j^2$
\begin{equation}\label{E:3.3.7}
\mathcal A_j(x)\leq 2\,.
\end{equation}
Using the fact that $r_j-r_{j+1}<r_1\leq 1$, we then conclude that
\begin{equation}\label{E:3.3.8,9}
\big|m_j(x)\big|\leq
\begin{cases}
\text{const.}(\alpha)\,\frac{1}{(r_j-r_{j+1})^2} &\quad \text{for } x\in A_j^1\,;\\
 & \\
2(\alpha+2)\,e^{-\alpha R_j} &\quad\text{for } x\in A_j^2\,.
\end{cases}
\end{equation}

Now let 
\begin{equation}\label{E:3.3.10}
\tilde m_j(x)=m_j(x)\cdot \omega(x)\,,
\end{equation}
and note that we find
\begin{equation}\label{E:3.3.13}
\big|\tilde m_j(x)\big|\leq  \text{const.}(\alpha)\cdot\left(1+\frac{\omega_j}{(r_j-r_{j+1})^2}\right)
\end{equation}
Then take in Lemma~\ref{L:B} 
\begin{equation}\label{E:3.3.14}
-E=\frac{\alpha^2}{2}\,.
\end{equation}
Since $H_0\geq 0$, \eqref{E:3.1.4} holds true, so that from \eqref{E:3.3.10} and \eqref{E:3.3.13} we conclude:
\begin{equation}\label{E:3.3.15}
\dlab \psi_E,|m_j|\psi_E\drab \leq \text{const.}(\alpha) \left(1+\frac{\omega_j}{(r_j-r_{j+1})^2}\right) \,
\int_{\bar A_j} \tilde\rho_\infty(x)\,dx\,.
\end{equation}

Let $K$ be a compact set, $K\subset\Omega$, and let $C(\alpha,K)=\inf_{x\in K} e^{-\alpha d_M(x)}>0$. Then,
as in the proof of Theorem~\ref{T:ess_sa_1}, one can find $N\geq 1$ an integer such that $K\subset \Omega_{r_N,R_N}$,
which then implies that
\begin{equation}\label{E:3.48}
\left(\sum_{j=N}^\infty \frac{1}{1+\frac{\omega_j}{(r_j-r_{j+1})^2}}\right) \int_K \big|\psi_E(x)\big|^2\rho_\infty(x)\,dx
\leq C(\alpha,K) \|\tilde\rho_\infty\|_{L^1}<\infty\,.
\end{equation}

But from \eqref{E:2.2.6} we know that 
\begin{equation*}
\sum_{j=N}^\infty\Big( \frac{\omega_j}{(r_j-r_{j+1})^2}\Big)^{-1}=\infty\,, 
\end{equation*}
which yields
\begin{equation*}
\sum_{j=N}^\infty \frac{1}{1+\frac{\omega_j}{(r_j-r_{j+1})^2}}=\infty\,.
\end{equation*}
Indeed, if we consider a sequence of positive numbers 
$(a_j)_{j\geq 1}\subset (0,\infty)$, then:
\begin{equation*}
\sum_{j=1}^\infty \frac{1}{a_j}=\infty\quad\Longrightarrow\quad \sum_{j=1}^\infty \frac{1}{1+a_j}=\infty\,.
\end{equation*}
This implication can be proven by considering $J=\{j\geq 1\,|\, a_j<1\}$. If $J$ is infinite, then note that
\begin{equation*}
\frac{1}{1+a_j}>\frac12\quad\text{for } j\in J\,,
\end{equation*}
and the conclusion follows. If $J$ is finite, then there exists $j_0\geq 1$ such that $a_j\geq 1$ for all $j\geq j_0$.
In this case, 
\begin{equation*}
\sum_{j=1}^\infty \frac{1}{1+a_j}\geq \sum_{j=j_0}^\infty \frac{1}{1+a_j}\geq \frac12\sum_{j=j_0}^\infty \frac{1}{a_j}=\infty\,,
\end{equation*}
and again the conclusion follows.

We have thus shown that $\int_K \big|\psi_E(x)\big|^2\,dx=0$, and the proof of Theorem~\ref{T:2}(iii) 
is completed.
\end{proof}

\section{Criteria in terms of the Euclidean metric}\label{S:5}

We turn now to our second type of results. The setting is as follows: Let $\Omega$ be bounded and its boundary consist 
of a finite number of well separated parts, as described in the hypotheses of the next theorem:
\begin{theorem}\label{T:3}
Let $\Omega$ be bounded, and $\partial\Omega=\bigcup_{j=1}^p \Gamma_j$ with $p<\infty$ and 
\begin{equation}\label{gammajgammak}
d(\Gamma_j,\Gamma_k)\geq d_0>0\quad\text{for all}\,\,\, j\neq k\,.
\end{equation}
For  $\nu>0$, let
\begin{equation}\label{deltaj}
\Gamma_{j,\nu}=\big\{x\in\Omega\,\big|\, \delta_j(x):=d(x,\Gamma_j)\leq\nu\big\}\quad\text{and}\quad v_j(\nu)=\text{vol} (\Gamma_{j,\nu})\,.
\end{equation}
Assume that
\begin{equation}\label{Da}
\mathbb D(x)\leq a(x)\mathds 1
\end{equation}
and that there exists $0<\nu_0  \leq  \min\{1, \frac{d_0}{4}\}$, $\varepsilon>0$, $\eta_j<1$ and $K<\infty$ such that for each $j=1,2,...p$, on $\Gamma_{j,\nu_0}$ either

i. 
\begin{equation}\label{ai}
a(x) \leq K\delta_j(x)^{\beta_j}, \quad \beta_j \geq 2
\end{equation}
and there exists $L<\infty$ such that
\begin{equation}\label{roi}
\rho_{\infty}(x) \leq
\begin{cases}
e^{L\delta_j(x)^{1-\frac{\beta_j}{2}}} &\quad \text{if } \quad  \beta_j >2\,;\\
 & \\
\delta_j(x)^{-L} &\quad \text{if} \quad  \beta =2\,,
\end{cases}
\end{equation}

or

ii. For $0<\nu\leq\nu_0$  $v_j$ is continuous,  $v_j(\nu)\leq K\nu^{1-\eta_j}$
and for $x\in\Gamma_{j,\nu_0}$ 
\begin{equation}\label{E:product_j}
a(x)\rho_\infty(x)\leq K\delta_j(x)^{1+\eta_j}\left(\ln \frac{1}{\delta_j(x)}\right)^{1-\varepsilon}\,.
\end{equation}
Then $H_0$
is stochastically complete.
\end{theorem}

\begin{remark}
Note that if we assume that the $\Gamma_j$ are $ C^2$-smooth, then the $\epsilon$-neighborhood Theorem
implies that the assumptions above are satisfied since the functions $v_j$ are $C^1$-smooth and 
$v_j(\nu)\sim C_j \nu^{\kappa_j}$ where $\kappa_j$ is the co-dimension of $\Gamma_j$. For more details,
see for example Theorem 2.2 in \cite{H}.
\end{remark}

\begin{proof}

The proof is similar to previous ones, but with a different choice for $f$ in the Basic Inequality Lemma \ref{L:B}. 
Let 
$\Gamma_{\nu}=\{x\in\Omega\,|\,\delta(x)\leq\nu\}$,
and  (recall that $\nu \leq \nu_0 \leq 1$) $B_{\nu}=\{x\in\Omega\,|\, \nu^2\leq\delta(x)\leq\nu\}$. Notice that $B_\nu\subset\Gamma_\nu$.
Since $\nu \leq \nu_0 \leq \frac{d_0}{4}$, from \eqref{gammajgammak} we find that for $j \neq k$
\begin{equation}\label{E:3.4.1}
d (\Gamma_{j,\nu},\Gamma_{k,\nu})\geq\frac{d_0}{2}\,,
\end{equation}
and on $\Gamma_{j,\nu_0}$ (see  \eqref{deltaj} for the definition of $\delta_j(x)$)
\begin{equation}\label{E:3.4.2}
\delta(x)=\delta_j(x)\,.
\end{equation}
For each $j\geq 1$, let
\begin{equation}\label{E:3.4.3}
B_{j,\nu}=\{x\in\Omega\,|\, \nu^2\leq\delta_j(x)\leq\nu\}\,,
\end{equation}
and note that from \eqref{E:3.4.1} we see 
\begin{equation}\label{E:3.4.4}
d(B_{j,\nu},B_{k,\nu})\geq\frac{d_0}{2}\,.
\end{equation}

We now choose the function $f_{\nu}=e^g\phi_{\nu}$ to be used in \eqref{E:3.1.5}. We begin with $\phi_{\nu}$.
 For $0<\nu \leq \nu_0$, 
\begin{equation}\label{E:3.4.6}
\phi_\nu(x)=\varphi_\nu(\delta(x))\,,
\end{equation}
with
\begin{equation}\label{E:3.4.7}
\varphi_\nu(t)=
\begin{cases}
0 & \quad\text{for } t\in(0,\nu^2) \\
1-\frac{\ln \frac{\nu}{t}}{\ln \frac{1}{\nu}} &\quad\text{for } t\in[\nu^2,\nu] \\
1 & \quad\text{for } t>\nu
\end{cases}
\end{equation}
Note that with this choice we have that for all $x$ such that $\delta(x) \geq \nu$
\begin{equation}\label{E:3.4.8}
\phi_\nu(x)=1.
\end{equation}
A direct calculation leads to ($t \neq \nu, \nu^2$)
\begin{equation}\label{E:3.4.9}
\varphi'_\nu(t)=
\begin{cases}
\frac{1}{t\ln \frac{1}{\nu}} & \quad\text{for } t\in (\nu^2,\nu) \\
0 & \quad\text{otherwise}
\end{cases}
\end{equation}
which in particular implies that
\begin{equation}\label{E:3.4.10}
\supp \phi_\nu\subset B_\nu\,.
\end{equation}
We turn now on the choice of $g(x)$. Let $J \subset \{1,2,...,p\}$ the set of $j$'s for which  alternative i. of the theorem 
holds true, $\mathcal G \in \mathcal C^1(0,\infty)$ ,   $|\mathcal G'(t)| \leq  \frac{4}{\nu_0}$  and 
\begin{equation}\label{g}
\mathcal G(t)=
\begin{cases}
1 &\quad \text{if } \quad t\in (0,
\frac{\nu_0}{2})\,;\\
 & \\
0 &\quad \text{if} \quad  t\geq \nu_0\,.
\end{cases}
\end{equation}
Fix $L<\infty$ and take
\begin{equation}
g(x)= -\sum_{j \in J}g_j(x)
\end{equation}
where
\begin{equation}\label{gj}
g_j(x)=
\begin{cases}
\mathcal G(L\delta_j(x)^{(1-
\frac{\beta_j}{2})}) &\quad \text{if } \quad \beta_j >2\,;\\
 & \\
\mathcal G(L \ln \frac{1}{\delta_j(x)})&\quad \text{if} \quad  \beta_j =2\,.
\end{cases}
\end{equation}
Notice that $g_j$ have disjoint supports. From the definition of $g_j$ and \eqref{ai}
\begin{equation}\label{gradMgj}
\big|\nM g_j(x)\big|_{_M}^2 \leq \text{const}(\nu_0, L, \beta_j)
\end{equation}
hence choosing $|E|$ sufficiently large in the Basic Inequality Lemma and using $H_0 \geq 0$, condition \eqref{E:3.1.4} is satisfied.

From \eqref{E:3.4.1} (in what follows $m_{j, \nu}$ equals  $|m_{\nu}|$ for $x \in B_{j,\nu}$ and vanishes otherwise)
\begin{equation}\label{E:3.4.12}
\dlab\psi_E,|m_{\nu}|\psi_E\drab \leq \|\psi_E\|_\infty^2 \sum_{j=1}^p 
\int_{B_{j,\nu}} m_{j,\nu}(x)\rho_\infty(x)\,dx\,.
\end{equation}
and  we are left with estimating $m_{j,\nu}(x)$.  Consider first the case $j \notin J$. Since by construction $g(x)=0$ on 
$B_{j,\nu}$, from \eqref{E:3.1.7}, \eqref{E:3.4.6} and \eqref{E:3.4.9}:
\begin{equation}\label{E:3.4.11}
\begin{aligned}
m_\nu(x) 
&= \big|\nM\phi_\nu(x)\big|_{_M}^2=\sum_{j,k} \mathbb D_{j,k}(x) \frac{\partial \phi_\nu}{\partial x_j}(x)\frac{\partial \phi_\nu}{\partial x_k}(x)\\
&\leq a(x)\big|\nabla \phi_\nu(x)\big|^2=a(x)\big|\varphi'_\nu(\delta(x))\big|^2 \big|\nabla\delta(x)\big|^2\\
&\leq \frac{a(x)}{\big(\ln\frac{1}{\nu}\big)^2\delta_j(x)^2}\,
\end{aligned}
\end{equation}
so that
\begin{equation}
\dlab\psi_E,|m_{j,\nu}|\psi_E\drab \leq \frac{\|\psi_E\|_\infty^2}{(\ln\frac{1}{\nu})^2}
\int_{B_{j,\nu}} \frac{a(x)\rho_\infty(x)}{\delta_j(x)^2}\,dx\,.
\end{equation}
From \eqref{E:product_j}, an integration by parts implies that
\begin{equation}\label{E:3.4.13}
\begin{aligned}
\int_{B_{j,\nu}} \frac{a(x)\rho_\infty(x)}{\delta_j(x)^2}\,dx
&\leq K \left(\ln \frac{1}{\nu^2}\right)^{1-\varepsilon}\int_{\nu^2}^\nu t^{-1+\eta_j}\,dv_j(t) \\
& =K \left(\ln \frac{1}{\nu^2}\right)^{1-\varepsilon} \bigg[\nu^{-1+\eta_j} v_j(\nu)
-\nu^{2(-1+\eta_j)} v_j(\nu^2)\\
&\qquad\qquad\qquad\qquad+(1-\eta_j)\int_{\nu^2}^\nu t^{\eta_j-2}v_j(t)\,dt \bigg]
\end{aligned}
\end{equation}
From our assumptions on $v_j$ we conclude that
\begin{equation*}
\nu^{-1+\eta_j} v_j(\nu)\,, \nu^{2(-1+\eta_j)} v_j(\nu^2) \leq K\,,
\end{equation*}
while
\begin{equation*}
\int_{\nu^2}^\nu t^{\eta_j-2}v_j(t)\,dt \leq K \int_{\nu^2}^\nu t^{-1}\,dt= K\,\ln\frac{1}{\nu}\,.
\end{equation*}
Combining everything above yields
\begin{equation}\label{E:3.4.15}
\int_{B_{j,\nu}} \frac{a(x)\rho_\infty(x)}{\delta_j(x)^2}\,dx\leq 8 K^2 \left(\ln \frac{1}{\nu}\right)^{2-\varepsilon} \,,
\end{equation}
and hence
\begin{equation}\label{E:3.4.16}
\dlab\psi_E,m_{j,\nu}\psi_E\drab \leq 8K^2 \|\psi_E\|_\infty^2 \left(\ln \frac{1}{\nu}\right)^{-\varepsilon}\,.
\end{equation}

Consider now the case $j \in J$. In this case (use $\big|(\nM u)\cdot (\nM v)\big| \leq \big| \nM v \big|_{_M} \big| \nM u \big|_{_M}$)
\begin{equation}
m_{j,\nu}(x) \leq e^{2g_j(x)}(\big| \nM \phi_{\nu}(x)\big|_{_M}^2+2\big| \nM \phi_{\nu}(x)\big|_{_M} \big| \nM g_j(x)\big|_{_M}),
\end{equation}
so using \eqref{E:3.4.11}, \eqref{gradMgj}, \eqref{gj} and \eqref{roi} one obtains
\begin{equation}\label{mjrho}
m_{j,\nu}(x)\rho_{\infty}(x) \leq  \text{const}(\nu_0, L, \beta_j) e^{g_j(x)}.
\end{equation}
From \eqref{gj}, on $B_{j,\nu}$
\begin{equation}
e^{g_j(x)} \leq
\begin{cases}
e^{-L\nu^{1-\frac{\beta_j}{2}}} &\quad \text{if } \quad \beta_j >2\,;\\
 & \\
e^{-L\ln\frac{1}{\nu}}&\quad \text{if} \quad  \beta_j =2\,
\end{cases}
\end{equation}
which together with \eqref{mjrho} implies
\begin{equation}
\lim_{\nu \searrow 0}\int_{B_{j,\nu}} m_{j,\nu}\rho_\infty(x) dx=0,
\end{equation}
and hence
\begin{equation}\label{E:3.4.17}
\lim_{\nu \searrow 0}\dlab\psi_E,m_{j,\nu}\psi_E\drab =0\,.
\end{equation}
Summing up \eqref{E:3.4.16} and \eqref{E:3.4.17}, for all $j=1,2,...,p$:
\begin{equation}\label{E:3.4.18}
\lim_{\nu \searrow 0}\dlab\psi_E,m_{j,\nu}\psi_E\drab =0\,.
\end{equation}

Now let $K\subset\Omega$ be a compact set, $E<0$, and $\nu>0$ sufficiently small so that 
$K\subset \{x\in\Omega\,|\, \delta(x)>\nu\}$. Then, from \eqref{E:3.1.5} and \eqref{E:3.4.18},
we obtain
\begin{equation*}
\int_K |\psi_E(x)|^2\rho_\infty(x)\,dx=0\,,
\end{equation*}
thus concluding the proof.
\end{proof}

Finally, we close with another result, this time a concrete criterion for essential self-adjointness.

\begin{theorem}\label{T:4}
Let $\Omega$, $\Gamma_j$, $\Gamma_{j,\nu}$ be as in Theorem~\ref{T:3}. 
Assume that there exists $\nu_0 >0$ such that for each $j=1,...,p$, 
either\\
(i) there exist $ M<\infty $ such that (see  \eqref{deltaj} for the definition of $\delta_j(x)$):
\begin{equation}\label{E:2.2.14}
\mathbb D(x)\leq M\delta_j(x)^2 \mathds 1\qquad\text{for all }x\in\Gamma_{j,\nu_0}\,,
\end{equation}
or\\
(ii) $\Gamma_j$ is a $C^2$ submanifold of $\mathbb R^d$ of dimension $d_j$, $0\leq d_j\leq d-1$, and for all $x\in \Gamma_{j,\nu_0}$
\begin{equation}\label{E:2.2.15}
D_{j,-}\delta_j(x)^{\beta_j}\mathds 1 \leq \ID(x)\leq D_{j,+}\delta_j(x)^{\beta_j}\mathds 1, \quad
\rho_{j,-}\delta_j(x)^{\gamma_j} \leq  \rho_\infty(x)\leq \rho_{j,+}\delta_j(x)^{\gamma_j}
\end{equation}
with
\begin{equation}\label{E:2.2.16}
\beta_j<2,\,\,\, \gamma_j\in\IR,\quad\text{and}\quad
 \frac{D_{j,-}\rho_{j,-}}{D_{j,+}\rho_{j,+}}   \left(\frac{\beta_j+\gamma_j+d-d_j-2}{2-\beta_j}\right)^2\geq1\,.
\end{equation}
Then $H_0$ is essentially self-adjoint.
\end{theorem}

Here we have to obtain first the "Hardy barrier", $B(x)$, appearing in Basic Inequality Lemma, and then to choose an appropriate $f$ there. A finite
number of positive constants $\nu_s >0$, $C_s<\infty $ will appear during the proof; if not otherwise stated, they depend only 
upon $H_0$ and the geometry of $\Omega$. 

Let $J$ be the set of $j$'s for which the alternative (ii) of the theorem holds true. 

\begin{lemma}[Hardy inequality]\label{L:H}
For  $\nu_1>0$ sufficiently small there exist $ 0\leq C_1, C_2<\infty $, $\nu_1 C_2  \leq 1$ , $\mathscr H(x) \geq 0$ such that for $\varphi \in C_0^{\infty}(\Omega)$
\begin{equation}
h_0[\varphi,\varphi]\geq  -C_1 \dlab\varphi, \varphi \drab + \int_{\Omega}\mathscr H(x)| \varphi(x) |^2 \rho_{\infty}(x)\,dx
\end{equation}
and for $x\in \Gamma_{j,\nu_1/2}$, $j\in J$:
\begin{equation}\label{E:H}
\mathscr H(x)=\frac{D_{j,-} \rho_{j,-}}{4\rho_{j,+}} (\beta_j +\gamma_j +d -d_j -2)^2\delta_j(x)^{\beta_j -2} (1-C_2 \delta_j(x))
\end{equation}
\end{lemma}

Before starting the proof of this lemma, we need to record two 
results. The first one gives the properties of  $\delta(x)$ near $\Gamma_j$, $j \in J$
\begin{lemma}\cite[Lemma 6.1 and Lemma 6.2]{Br}\label{L:4.5}
Let  $j\in J$.
Then there exist $\nu_2>0$ and $C_3<\infty$ such that
$\delta=\delta_j\in C^2(\Gamma_{j,\nu_2})$, and 
\begin{equation}\label{E:4.66}
|\nabla\delta(x)|=1\quad\text{and}\quad \left|\Delta\delta(x)-\frac{d-d_j-1}{\delta(x)}\right|\leq C_3
\quad\text{for all } x\in\Gamma_{j,\nu_2}.
\end{equation}
\end{lemma}

The second result describes the so called "vector field/ground state representation approach"  to Hardy inequalities
\begin{lemma}\cite[Theorem 4.1]{Br}; see also \cite{BFT},\cite{Mi},\cite{Lu}.\label{L:4.6}
Let $X \in C^1(\Omega;\IR^d)$ be a differentiable real vector field on $\Omega$. Then
\begin{equation}\label{E:4.67}
h_0[\varphi, \varphi]  \geq  \int_{\Omega}\left(\nabla\cdot X(x)-X(x)\cdot \big(\rho_\infty(x) \ID(x)\big)^{-1}X(x)\right)\,
\big|\varphi(x)\big|^2\,dx\,,
\end{equation}
for all $\varphi\in C_0^\infty(\Omega)$.
\end{lemma}

\begin{proof}[Proof of Lemma~\ref{L:H}]
The proof consists in making an appropriate  choice of the vector field $X$ and plugging it in \eqref{E:4.67}.
For $j\in J$ and $x \in \Gamma_{j, \nu_3}$, $\nu_3= \min\{1, \nu_0,\nu_2\}$, let
\begin{equation}\label{E:Xj}
X_j(x)= h_j\delta_j(x)^{\gamma_j+\beta_j-1}\,\nabla \delta_j(x)\,,
\end{equation}
with the constant $h_j$ to be determined later.
Given these choices, we set
\begin{equation}\label{E:X}
X(x)=\sum_{j\in J} X_j(x)\psi_j(x)\,,
\end{equation}
where $\psi_j\in C^1(\Omega)$ are cut-off functions, $0\leq \psi_j\leq 1$ on $\Omega$, and
\begin{equation}\label{E:psij}
\psi_j(x)=
\begin{cases}
1 & \quad\text{for } x\in\Gamma_{j,\frac{\nu_3}{2}}\\
0 & \quad\text{for } x\not\in\Gamma_{j,\nu_3} 
\end{cases}\,,
\end{equation}

To estimate the integrand on the rhs of \eqref{E:4.67} with the $X$ chosen above, we first notice that,
for $x\in\Omega$, the following expressions hold:
\begin{equation}\label{E:XDX}
\begin{aligned}
\nabla\cdot X-X\cdot \big(\rho_\infty \ID\big)^{-1}X
&=\begin{cases}
\nabla\cdot \big(X_j\psi_j\big)-\big(X_j\psi_j\big)\cdot \big(\rho_\infty \ID\big)^{-1}\big(X_j\psi_j\big)
&\quad\text{on each } \Gamma_{j,\nu_3}\,,\\
0 &\quad \text{otherwise}\,,
\end{cases}\\
&=\begin{cases}
\nabla\cdot X_j-X_j\cdot \big(\rho_\infty \ID\big)^{-1}X_j
&\quad\text{on each } \Gamma_{j,\frac{\nu_3}{2}}\,,\\
\nabla\cdot \big(X_j\psi_j\big)-\big(X_j\psi_j\big)\cdot \big(\rho_\infty \ID\big)^{-1}\big(X_j\psi_j\big)
&\quad\text{on each } \Gamma_{j,\nu_3}\setminus \Gamma_{j,\frac{\nu_3}{2}}\,,\\
0 &\quad \text{otherwise.}
\end{cases}
\end{aligned}
\end{equation}

On $\Gamma_{j,\nu_3}$,
 with the hypotheses on $\rho_{\infty}$ and $\mathbb D$ one obtains (see \eqref{E:2.2.15},  \eqref{E:Xj})
\begin{equation}\label{E:divX}
\nabla\cdot X_j
=
h_j \left((\beta_j+\gamma_j-1)\delta_j^{\beta_j+\gamma_j-2}\big|\nabla\delta\big|^2+
    \delta^{\beta_j+\gamma_j-1}\Delta\delta_j\right),
\end{equation}
\begin{equation}\label{E:XX}
X_j\cdot \big(\rho_\infty \ID\big)^{-1}X_j \leq (\rho_{j,-} D_{j,-})^{-1} h_j^2\delta_j^{\gamma_j+\beta_j-2}\big|\nabla\delta_j\big|^2\,.
\end{equation}

If we further recall that $\psi_j\equiv 1$ on $\Gamma_{j,\frac{\nu_3}{2}}$, we obtain from \eqref{E:divX} and \eqref{E:XX} that for $x\in \Gamma_{j,\frac{\nu_3}{2}}$, $j \in J$
\begin{equation*}
\nabla\cdot X-X\cdot \big(\rho_\infty \ID\big)^{-1}X
\geq 
\delta_j^{\gamma_j+\beta_j-2}\,\left[h_j(\gamma_j+\beta_j+d-d_j-2)-\frac{h_j^2}{\rho_{j,-} D_{j,-}} 
+h_j \delta_j\,\left(\Delta\delta-\frac{d_j-1}{\delta_j}\right)\right]
\end{equation*}

We now fix $h_j$ by maximizing the quadratic polynomial given by the first two terms in the square bracket. This leads to
\begin{equation*}
h_j = \rho_{j,-} D_{j,-} \frac{\beta_j +\gamma_j+d-d_j-2}{2 }\,.
\end{equation*}
Hence, for $x\in \Gamma_{j,\frac{\nu_3}{2}}$, $j \in J$
\begin{equation*}
\begin{aligned}
\nabla\cdot X-X\cdot \big(\rho_\infty \ID\big)^{-1}X
&\geq \rho_{j,-} D_{j,-} \delta_j^{\gamma_j+\beta_j-2}\left[\left( \frac{\beta_j +\gamma_j +d -d_j -2}{2}\right)^2\right. \\
&\qquad\qquad\qquad\qquad\left.+\delta_j\left(\frac{2}{\beta_j +\gamma_j +d_j -2}\right)
\left(\Delta\delta-\frac{d_j-1}{\delta_j}\right)\right]
\end{aligned}
\end{equation*}
which together with \eqref{E:4.66} leads to (notice that due to \eqref{E:2.2.16}, $\beta_j +\gamma_j +d -d_j -2 \neq0$)
\begin{equation}\label{E:XXgeq}
\begin{aligned}
\nabla\cdot X-X\cdot \big(\rho_\infty \ID\big)^{-1}X
&\geq \rho_{j,-} D_{j,-} \delta_j^{\gamma_j+\beta_j-2}\left( \frac{\beta_j +\gamma_j +d -d_j -2}{2}\right)^2\times\\
&\qquad\qquad\qquad\times\left[ 1-
\delta_j \Big| \frac{2}{\beta_j +\gamma_j +d_j -2}\Big| C_3\right],
\end{aligned}
\end{equation} 
for all $x\in \Gamma_{j,\frac{\nu_3}{2}}$, $j \in J$.

Take now
\begin{equation} \label{E:nu1}
0<\nu_1 \leq \min \{\nu_3, \Big| \frac{\beta_j +\gamma_j +d -d_j -2}{2}\Big| C_3\},
\end{equation}
\begin{equation}\label{E:C2}
C_2 =\Big| \frac{2}{\beta_j +\gamma_j +d_j -2}\Big| C_3,
\end{equation}
\begin{equation}\label{E:C1}
C_1(\nu_1)= \sup_{x\in \cup_{j\in J}\left(\Gamma_{j,\nu_3}\setminus \Gamma_{j,\frac{\nu_1}{2}}\right)}
\frac{\big|\nabla\cdot X(x)-X(x)\cdot \big(\rho_\infty(x) \mathbb D(x))^{-1}X(x)\big|}{\rho_{\infty}(x)} .
\end{equation}
From \eqref{E:XX}, \eqref{E:divX},  \eqref{E:XDX} and the properties of $\psi_j$ it follows immediately that that for all $\nu_1 >0$, $C_1(\nu_1)<\infty$.

With these choices Lemma \ref{L:H} follows from Lemma \ref{L:4.6}, \eqref{E:XXgeq} and \eqref{E:2.2.15}.
\end{proof}

\begin{proof}[Proof of Theorem~\ref{T:4}]
Fix $\nu_1$ in Lemma \ref{L:H} small enough such that besides \eqref{E:nu1},  $\nu_1 \leq e^{-e}$ and in addition
for all $j \in J$ and $t \in (0, \nu_1)$
\begin{equation} \label{E:tlnt}
\frac{1}{2} \geq \frac{1}{(2-\beta_j)\ln\frac{1}{t}} \geq C_2t.
\end{equation}
From Lemma \ref{L:H} we get that the condition \eqref{E:h} in Lemma \ref{L:B} with $E_0=-C_1$ and 
\begin{equation}\label{E:B}
B(x)=\begin{cases}
\frac{D_{j,-}\rho_{j,-}}{4\rho_{j,+}} (\beta_j +\gamma_j +d -d_j -2)^2\delta_j(x)^{\beta_j -2} (1-C_2 \delta_j(x)) &
 \quad\text{for } x\in\Gamma_{j,\frac{\nu_1}{2}}, j\in J\\
0 & \quad\text{othervise} 
\end{cases}\,,
\end{equation}
is satisfied. Further, in order to apply Lemma \ref{L:B} we have to choose $g(x)$, $E >E_0$ such that \eqref{E:3.1.4} holds true.

For $j\in J$, $t\in (0,e^{-e})$, let:
\begin{equation}\label{E:Gj}
G_j(t)= \ln t^{\frac{2-\beta_j}{2}}+\frac{1}{2}\ln \ln \frac{1}{t}.
\end{equation}
For later use, notice that
\begin{equation}\label{E:expG}
t^{\beta_j}e^{2G_j(t)}= t^2\ln \frac{1}{t},
\end{equation}
and
\begin{equation}\label{E:dGj}
\frac{d}{dt}G_j(t)=\frac{2-\beta_j}{2t}\left(1-\frac{1}{(2-\beta_j)\ln\frac{1}{t}}\right).
\end{equation}
We  set
\begin{equation}\label{E:gTh4}
g(x)= \sum_{j \in J}G_j(\delta_j(x))\psi_j(x),
\end{equation}
where $\psi_j$ are given by \eqref{E:psij} with $\nu_3$ replaced by $\nu_1$.
By construction, the terms on the r.h.s. of \eqref{E:gTh4} have disjoint supports, 
\begin{equation}\label{E:supgradg}
\sup_{x \in \Omega \setminus (\cup_{j\in J} \Gamma_{j,\frac{\nu_1}{2}})}|\nM g(x)|_{_M} \leq C_4,
\end{equation}
and for any compact $K \subset \Omega$
\begin{equation}\label{E:infexpg}
\inf_{x\in K}e^{2g(x)} \geq C(K) >0.
\end{equation}

Now, if
\begin{equation}\label{E:EE0}
\frac{E_0-E}{2} \geq C_4^2,
\end{equation}
\eqref{E:3.1.4} holds true. Indeed, on $\Omega \setminus (\cup_{j\in J} \Gamma_{j,\frac{\nu_1}{2}})$      
  from \eqref{E:supgradg}, \eqref{E:EE0} 
\begin{equation}\label{E:gradMg} 
|\nM g|_{_M}^2 \leq C_4^2 \leq \frac{|E_0-E|}{2} \leq \frac{|E_0-E|}{2} +B.
\end{equation}

On the other hand on $  \Gamma_{j,{\frac{\nu_1}{2}}}$, from the definition of $g$, $|\nabla\delta_j|=1$, $\psi_j =1$, 
\eqref{E:2.2.15}, \eqref{E:dGj}, \eqref{E:tlnt}, \eqref{E:B} and \eqref{E:2.2.16}:

\begin{equation}\label{E:gradMgB}
\begin{aligned}
|\nM g|_{_M}^2
&= \nabla g\cdot \mathbb D\nabla g 
 \leq D_{j,+}(G'(\delta_j))^2\delta_j^{\beta_j}\\
&\leq D_{j,+}\left(\frac{2-\beta_j}{2}\right)^2  \delta_j^{\beta_j-2}\left(1-\frac{1}{(2-\beta_j)\ln\frac{1}{\delta_j}}\right)^2 \\
& \leq\frac{D_{j,+}\rho_{j,+}}{D_{j,-}\rho_{j,-}}\left(\frac{2-\beta_j}{\beta_j +\gamma_j +d-d_j-2}\right)^2
\frac{\left(1-\frac{1}{(2-\beta_j)\ln\frac{1}{\delta_j}}\right)^2}{1-C_2\delta_j}B\\
&\leq B\leq B+\frac{|E-E_0|}{2}.
\end{aligned}
\end{equation}
According to the general scheme it remains to choose $\phi$ and estimate the r.h.s. of \eqref{E:3.1.5}. As expected the choice of $\phi$ is very similar to the one in the proof of Theorem \ref{T:ess_sa_1}. More precisely, let for $l=1,2,...$, $r_l=r_12^{1-l}$, $r_1 \leq \frac{\nu_1}{2}$,  $k_l(t)$ as in \eqref{E:3.2.2,3} and
\begin{equation}\label{E:phil}
\phi(x)=k_l(\delta(x)).
\end{equation}
Further, if
\begin{equation}\label{E:Bjl}
B_{j,l}=\{x\in \Gamma_{j, \frac{\nu_1}{2}}\,\big|\, r_{l+1} <\delta_j(x)<r_l\}
\end{equation}
\begin{equation}\label{E:Bl}
B_l =\bigcup_{j=1}^p B_{j,l}
\end{equation}
 then the $B_{j,l}$ are disjoint and
 \begin{equation}\label{E:gradphi}
 \supp \nabla\phi_l \subset B_l,\quad |\nabla \phi_l| \leq \frac{2}{r_{l+1}}.
 \end{equation}
 From \eqref{E:3.1.5}, \eqref{E:3.1.7}, \eqref{E:gradphi}, \eqref{E:EE0} and the fact that $B_{j,l}$ are disjoint one has:
 \begin{equation}\label{E:PsiEPsi}
\dlab\psi_E , e^{2g}\phi_l^2\psi_E\drab\leq \frac{2}{|E-E_0|} \sum_{j=1}^p\int_{B_{j,l}}m_{j,l}(x)|\Psi_{E}(x)|^2\rho_{\infty}(x)dx,
\end{equation}
\begin{equation}\label{E:mjl}
m_{j,l}(x)=e^{2g(x)}\big(2|\nM g(x)|_{_M} |\nM\phi_l(x)|_{_M}+|\nM \phi_l(x)|_{_M}^2 \big)\,.
\end{equation}
In estimating the r.h.s. of \eqref{E:PsiEPsi}, consider first the terms $j \in J$. Here  from the fact that $G_j(t)$ is increasing, \eqref{E:expG}, \eqref{E:gradMgB} and \eqref{E:gradphi}
:
\begin{equation}\label{exp2g}
e^{2g(x)} = e^{2G_j(\delta_j(x))}\leq r_l^{2-\beta_j}\ln\frac{1}{r_l}
\end{equation}
\begin{equation}\label{E:nablagphi}
|\nM g(x)|_{_M}\leq (D_{j,+}\delta_j(x)^{\beta_j})^{1/2}\frac{2-\beta_j}{2\delta_j(x)},\quad |\nM \phi_l(x)|_{_M} \leq \frac{2(D_{j,+}\delta_j(x)^{\beta_j})^{1/2}}{r_{l+1}}
\end{equation}
which gives
\begin{equation}\label{mjlJ}
m_{l,j}(x) \leq C_5l.
\end{equation} 
For $j\notin J$, since (see \eqref{E:gTh4}) $g(x)=0$, from \eqref{E:gradphi} , \eqref{E:mjl} and \eqref{E:2.2.14}:
\begin{equation}\label{mjlnonJ}
m_{l,j}(x) \leq C_6.
\end{equation}
Putting together \eqref{E:PsiEPsi}, \eqref{mjlJ} and \eqref{mjlnonJ} one gets:
\begin{equation}\label{E:PsiEPsil}
\dlab\psi_E , e^{2g}\phi_l^2\psi_E\drab\leq C_7 \frac{p}{|E-E_0|}l\int_{B_l}|\Psi_{E}(x)|^2\rho_{\infty}(x)dx.
\end{equation}

From this point on, the proof closely mimics the end of the proof of Theorem \ref{T:ess_sa_1}. Fix a compact $K\subset \Omega$. There exists an integer $L(K)$ such that
$K \subset \{x \in \Omega\big| \delta(x) < r_{L(K)}\}$. Taking into account that by construction $\phi \big|_K =1$, \eqref{E:infexpg} and \eqref{E:PsiEPsil} one finds
\begin{equation*}
C(K)\int_K|\Psi_{E}(x)|^2\rho_{\infty}(x)dx \leq \int_K|\Psi_{E}(x)|^2f_l(x)^2\rho_{\infty}(x)dx \leq C_7 \frac{p}{|E-E_0|}l\int_{B_l}|\Psi_{E}(x)|^2\rho_{\infty}(x)dx
\end{equation*}
which gives
\begin{equation*}
\frac{|E-E_0|C(K)}{2C_7}\frac{1}{l}\int_K|\Psi_{E}(x)|^22\rho_{\infty}(x)dx \leq \int_{B_l}|\Psi_{E}(x)|^2\rho_{\infty}(x)dx.
\end{equation*}

 Summing over $l $ from $L(K)$ to $N$ and taking onto account that $B_l $ are disjoint one obtains
 \begin{equation*}
\frac{|E-E_0|C(K)}{2C_7}\left(\sum_{l=L(K)}^N\frac{1}{l}\right)\int_K|\Psi_{E}(x)|^2\rho_{\infty}(x)dx \leq \|\Psi_E\|^2,
\end{equation*}
which in the limit $N\rightarrow \infty$ gives $\int_K|\Psi_{E}(x)|^2\rho_{\infty}(x)dx =0$, hence $\Psi_E=0$  and the application of Theorem \ref{ES} 
finishes the proof.

 \end{proof}
\section{Remarks and examples}\label{S:6}

\begin{remark}
In the example below  $M$ is not complete, $\text{diam}(M) <\infty$ and for $r$ sufficiently small $S_r = \{ x\in \Omega\, |\,\delta_M(x) \geq r\}$ is not compact, 
hence  Assumption \textbf{A} is not superfluous. 

Let $\Omega \subset \mathbb R^2$, 
$\Omega =\{x=(x_1,x_2)\,|\, x_1 \in \mathbb R, |x_2|<1\}$, $\mathbb D=(1-x_2^2)^{-1}\mathds 1$. By direct computation
\begin{equation}
\delta_M(x)= \int_{|x_2|}^1(1-u^2)^{1/2}du
\end{equation}
and for $|x_1|$ sufficiently large
\begin{equation}
d_M(x,0)=\int_0^1(1-u^2)^{1/2}du +\int_{|x_2|}^1(1-u^2)^{1/2}du.
\end{equation}
It follows that for  $r$ sufficiently small $\{x\in \Omega\, |\, x_2 =0\} \subset S_r$, hence $S_r$ cannot be compact.
\end{remark}

\begin{remark}
In Theorems \ref{T:3} and \ref{T:4} we restricted ourselves to bounded domains. As far as  $\partial \Omega$ is bounded, 
the extension to unbounded domains can be treated in the same manner by considering the point  at infinity as a part, $\Gamma_{\infty},$ of the 
 "boundary" of $\Omega$ and imposing conditions on $\mathbb D(x)$ and $\rho_{\infty}(x)$ as $|x| \rightarrow \infty$. For example by  slight modifications  
of the proof of Theorem \ref{T:3} and Theorem \ref{T:4} one obtains:
\begin{theorem}\label{T:5}
Suppose $\{x \in \mathbb R^d \,|\, |x| > R\} \subset \Omega$ for some $R<\infty$, $\partial \Omega$ is as in Theorem \ref{T:3},
$\partial \Omega \subset \{x \in \mathbb R^d \,|\, |x| < R/2\}$.

i. In addition to the conditions in Theorem \ref{T:3} assume that for $|x| >R$:
\begin{equation}\label{ainfty}
a(x) \leq K|x|^{\beta_{\infty}}, \quad \beta_{\infty} \leq 2,
\end{equation}
\begin{equation}\label{roinfty}
\rho_{\infty}(x) \leq
\begin{cases}
e^{L|x|^{1-\frac{\beta_{\infty}}{2}}} &\quad \text{if } \quad  \beta_{\infty} <2\,;\\
 & \\
|x|^{L} &\quad \text{if} \quad  \beta_{\infty} =2\,.
\end{cases}
\end{equation}
Then $H_0$
is stochastically complete.

ii.  In addition to the conditions in Theorem \ref{T:4} assume that for $|x| >R$ \eqref{ainfty} holds true.
Then $H_0$ is essentially self-adjoint.
\end{theorem}
\end{remark}

\begin{remark}
In a setting similar with ours, stochastic completeness and essential self-adjointness for $H_0$ have been recently studied by W. D. Robinson and 
A. Sikora \cite{RoSi1} for the particular case when the drift potential vanishes i.e. $\rho_{\infty} \equiv 1$, under the assumption that
$D_{j,k} \in W^{1,\infty}(\Omega)$. The smoothness conditions on $D_{j,k}$ are weaker than ours but adding the necessary technicalities our results
can be extended to the case  when $D_{j,k} \in W^{1,\infty}_{loc}(\Omega)$. 
The method in \cite{RoSi1} is completely different from ours and is based on the theory of Dirichlet forms and capacity estimates \cite{Eb,FOT, MR}.
Concerning the optimality question, when applied to the same geometry of $\Omega$ our condition on  $\mathbb D$ is weaker by a logarithmic factor 
(see Example \ref{E:1} below).
\end{remark}

\begin{remark}
When applying Theorem~\ref{T:3} to concrete cases, one has to compute $\eta_j$. This is easy if $\Gamma_j$ is $C^2$-smooth 
(see Remark 5.2 above), namely $\eta_j=1-\kappa_j$, where $\kappa_j$ is the co-dimension of $\Gamma_j$. So, for example,
if $\Gamma_j$ is a closed line segment, then $\eta_j=2-d$. But $\eta_j$ is also computable (often by hand) in other cases.
For example, consider 
\begin{equation*}
\Omega=\big\{(x,y)\in\IR^2\,\big|\, x^2+y^2<1\big\}\setminus \Gamma\,,
\end{equation*}
where, for some $N\geq2$,
\begin{equation*}
\Gamma=\bigcup_{n=N}^\infty I_n\,,\qquad I_n=\big\{(x,y)\in\IR^2\,\big|\, x=\tfrac1n\text{ and }|y|\leq\tfrac12\big\}\,.
\end{equation*}
In this case, a straightforward investigation shows that $\eta=\frac12$. 

In more general situations one can, for example, use the Minkowski dimension to compute $\eta_j$. We do not consider
such general cases in this paper, but rather restrict ourselves to the $C^2$-smooth case, where Theorem~\ref{T:3}
leads to sharp results. 

A detailed study of the effect
of roughness of the boundary in a related problem (namely the Markov uniqueness for the case $\rho_{\infty} \equiv 1$)
has recently been done in \cite{LR}.
\end{remark}

We give now some particular cases of Theorems \ref{T:3}, \ref{T:4}, \ref{T:5}. 

\begin{example}\label{E:1}
Suppose $\Omega$ is bounded and simply connected with $ C^2$ boundary.  Then 
$$
\text{vol} \{x\in \Omega \,\big|\, \delta(x) \leq \nu\} \leq \text{const.}\,\nu
$$
and \eqref{E:product_j} reduces to (see \eqref{Da})
\begin{equation}
a(x)\rho_\infty(x)\leq K\delta(x)\left(\ln \frac{1}{\delta(x)}\right)^{1-\varepsilon}\,.
\end{equation}
In particular if $\rho_{\infty} \equiv 1$, $H_0$
is stochastically complete for
 \begin{equation}
a(x)\leq K\delta(x)\left(\ln \frac{1}{\delta(x)}\right)^{1-\varepsilon}\,
\end{equation}
while Corollary 4.4 in \cite{RoSi1} gives the condition
\begin{equation}
a(x)\leq K\delta(x).
\end{equation} 
\end{example}

Consider now the case when $\Omega =\mathbb R^d \setminus \{ 0\}$ and assume that \eqref{ainfty}, \eqref{roinfty}  hold  true so that there 
is no obstruction to stochastic completeness from the 
neighborhood of infinity.

\begin{example}\label{E:2}
Suppose that in a neighborhood of the origin
\begin{equation}
a(x)\rho_\infty(x)\leq K |x|^{2-d}\left(\ln \frac{1}{|x|}\right)^{1-\varepsilon}\,, \quad \varepsilon >0.
\end{equation}
Then, taking into account that  $\text{vol}\{ x \,\big|\,|x| \leq \nu \}\sim \nu^d$, $H_0$ is stochastically complete.
\end{example}
Notice that in the last example, for $d\geq 2$ both $a(x)$ and $\rho_{\infty}(x)$ can blow up as $|x|\rightarrow 0$.

In the next example  we suppose that near $\partial \Omega$, $\mathbb D$ and $\rho_{\infty}$ have power-like behavior.
\begin{example}\label{E:3}
In Example \ref{E:1} suppose that for sufficiently small $\delta(x)$:
\begin{equation}
\mathbb D(x)=D\delta(x)^{\beta}\mathds 1, \quad \rho_{\infty}(x)=\rho \delta(x)^{\gamma}\,,\quad\text{with } D,\rho\in(0,\infty)\,.
\end{equation}
Then:

i. If either $\beta \geq 2$, or $\beta< 2$ and
\begin{equation}\label{selfbetagamma}
\left(\frac{\beta+\gamma-1}{2-\beta} \right)^2\geq 1,
 \end{equation}
then $H_0$ is essentially self-adjoint.
 
 ii. If either $\beta \geq 2$, $\gamma \in \mathbb R$ or
 \begin{equation}\label{scbetagamma}
 \beta +\gamma \geq 1,
 \end{equation}
 then $H_0$ is stochastically complete.
 \end{example}
 In particular if $\gamma=0$ (i.e. the drift potential is constant) then from i. it follows that $H_0$ is essentially self-adjoint
 for 
 \begin{equation}
 \beta \geq 3/2
 \end{equation}
 and stochastically complete for
 \begin{equation}
 \beta \geq 1\,,
 \end{equation}
 which is the generalization to higher dimension of known results in $d=1$ (see e.g. \cite{RoSi2}). 
 
 If $\beta=0$ then from \eqref{selfbetagamma} it follows that
 $H_0$ is essentially self-adjoint for
 \begin{equation}
 \gamma \in (-\infty,-1] \cup [3,\infty)
 \end{equation}
 and stochastically complete for
 \begin{equation}
 \gamma \geq 1.
 \end{equation}
 Notice that in this case the sets of $\gamma$s for which $H_0$ is essential self-adjoint, respectively stochastically complete, are different, and neither of them 
 includes the other.

 In a recent paper \cite{BP} (see also \cite{BL}) 
U. Boscain and D. Prandi made a detailed study of the Laplace-Beltrami operator on conic and anticonic two dimensional surfaces.
 Among other things they proved that the Laplace-Beltrami operator 
 \begin{equation*}
 \Delta_+\text{ on }M_+=\big\{ (x, \theta)\, \big|\,\, x\in (0, \infty),\,\, \theta \in \mathbb T \big\}
 \end{equation*} 
endowed with the metric 
\begin{equation*}
ds^2= dx^2 + x^{-2\alpha} d\theta^2\,,\quad\alpha\in\mathbb R
\end{equation*}
is essentially self-adjoint if and only if $\alpha \notin (-3,1)$, and stochastically complete if and only if 
$\alpha \leq -1$. The result covers in particular one of the well-known examples where essential self-adjointness does not imply stochastic completeness: 
the Grushin metric plane.
 
Taking  $\gamma=-\alpha$, $\beta=0$ in Example \ref{E:3} above one sees that the conditions for essential self-adjointness and stochastic completeness 
are the same in spite of the fact the operators are very different: in Example \ref{E:3}, $\Omega$ is an arbitrary simple connected, bounded   domain in $
\mathbb R^d$, $d \geq 1$, with smooth boundary, while $M_+$ is two dimensional. 
On the other hand, in Example \ref{E:3} we require that $\mathbb D$ is a multiple of unity while this is not the case for $\Delta_+$.
Still, this fact can be understood at the heuristic level. Indeed, for both essential self-adjointness and stochastic completeness the relevant part of the 
operator under consideration
is the one describing the one-dimensional ``motion'' along the normal to the boundary, and in both cases this has (in appropriate coordinates)
the same form, $ x^{\alpha}\partial_x x^{-\alpha}\partial_x$, $x \in (0,x_0)$ for some $x_0 >0$.


\end{document}